\documentclass[final,3p,times]{elsarticle}
\journal{***}
\usepackage{amssymb}
\usepackage{lipsum}
\usepackage{graphicx}%
\usepackage{multirow}%
\usepackage{amsmath,amssymb,amsfonts}%
\usepackage{amsthm}%
\usepackage{mathrsfs}%
\usepackage[title]{appendix}%
\usepackage{xcolor}%
\usepackage{textcomp}%
\usepackage{manyfoot}%
\usepackage{booktabs}%
\usepackage{algorithm}%
\usepackage{algorithmicx}%
\usepackage{algpseudocode}%
\usepackage{listings}
\usepackage[breaklinks=true,letterpaper=true,colorlinks,bookmarks=false]{hyperref}

\theoremstyle{thmstyleone}%
\newtheorem{theorem}{Theorem}
\theoremstyle{thmstyletwo}%
\newtheorem{remark}{Remark}%

\theoremstyle{thmstylethree}%

\usepackage[pagewise]{lineno}
\usepackage{caption}
\begin{document}
	
	\begin{frontmatter}
		
		
		
		\title{A tumor-immune model of chronic myeloid leukemia with optimal immunotherapeutic protocols}
		
		
		\author[First]{Haifeng Zhang}
		\ead{haifengzhang202212@163.com}
		
		\author[Second]{Changjing Zhuge}
		\ead{zhuge@bjut.edu.cn}
		
		\author[Third]{Jinzhi Lei\corref{cor1}}
		\ead{jzlei@tiangong.edu.cn}
		
		\cortext[cor1]{Corresponding author.}

		\address[First]{School of Mathematical Sciences, Jiangsu University, Zhenjiang, 212013, P. R. China.}
		\address[Second]{Beijing Institute for Scientific and Engineering Computing, Beijing University of Technology,  Beijing, 100124, P. R. China.}
		\address[Third]{School of Mathematical Sciences, Center for Applied Mathematics, Tiangong University, Tianjin, 300387, P. R. China.}
		
\begin{abstract}
The interactions between tumor cells and the immune system play a crucial role in cancer evolution. In this study, we explore how these interactions influence cancer progression by modeling the relationships among naive T cells, effector T cells, and chronic myeloid leukemia cells. We examine the existence of equilibria, the asymptotic stability of the positive steady state, and the global stability of the tumor-free equilibrium. Additionally, we develop a partial differential equation to describe the conditions under which the concentration of cancer cells reaches a level that allows for effective control of cancer evolution. Finally, we apply our proposed model to investigate optimal treatment strategies that aim to minimize both the concentration of cancer cells at the end of treatment and the accumulation of tumor burden, as well as the cost associated with treatment during the intervention period. Our study reveals an optimal therapeutic protocol using optimal control theory. We perform numerical simulations to illustrate our theoretical results and to explore the dynamic behavior of the system and optimal therapeutic protocols. The simulations indicate that the optimal treatment strategy can be more effective than a constant treatment approach, even when applying the same treatment interval and total drug input.
\end{abstract}

\begin{keyword}
	Mathematical model \sep cancer immunotherapy optimal control \sep Pontryagin's minimum principle
\MSC[2020] 34C60 \sep 34D20 \sep 93C95
	
	
\end{keyword}

\end{frontmatter}

\section{Introduction}
Leukemia refers to a group of blood cancers marked by a significant increase in abnormal blood cells, which pose serious risks to human health. Challenges remain while considerable progress has been made in controlling and eliminating leukemia cells. Cancer immunotherapy, which utilizes the immune system to combat cancer, has become an important approach to cancer treatment. Research has increasingly demonstrated that the immune system plays a crucial role in the development of myeloid leukemia \citep{Vago2020,Sawyers1999,Riether2015,Riether2022}, and that immune deficiencies can increase susceptibility to carcinogens \citep{Vesely2011}. Understanding how the immune system influences cancer development and progression presents one of the most significant and complex questions in immunology \citep{Schreiber2011}. In this study, we focus on chronic myeloid leukemia (CML), a specific type of leukemia characterized by the presence of Philadelphia chromosome, and we explore optimal immunotherapeutic protocols through a mathematical model.

To gain a better understanding of how the immune system influences cancer development, numerous mathematical models have been developed to explore the interactions between tumors and the immune system \citep{Ruan2021,Eftimie2011,Moore2004,Talkington2018,Liu2017,Bi2014,Clapp2015,Besse2018,Yang2015,Li2022,Cho2020,Lai2020,Shi2021,Han2020}. Kuznetsov et al. \citep{Kuznetsov1994} proposed a model that describes the interactions between effector cells and tumor cells. The dynamic behavior of this model includes several phenomena: the immunostimulation of tumor growth, the tumor's ability to ``sneak through" immune defenses, and the formation of a dormant tumor state. Kirschner and Panetta \citep{Kirschner1998} developed a mathematical model for tumor-immune dynamics that explains short-term oscillations in tumor size and long-term tumor relapse, further applying the model to explore the effects of immunotherapy. Dong et al. \citep{Dong2014} utilized a bilinear ordinary differential equation system to report that helper T cells are critical for sustaining long-term periodic oscillation in tumor-immune interactions. Shi et al. \citep{Shi2023} employed a tumor-immune system model to investigate the conditions leading to bistability and the long-term coexistence of tumor and T cells under anti-PD-1 treatment. Li et al. \citep{Li2021} proposed and analyzed a conceptual model that includes both tumor and immune cells, demonstrating that antigenicity significantly controls tumor development. d'Onofrion \citep{dOnofrio2005} studied a general family of models for tumor-immune interactions and examined cancer evolution dynamics following various treatment protocols. Ledzewicz and Schättler \citep{Ledzewicz2023} focused on a dynamic model combining immunotherapy with radiotherapy, investigating three medically realistic scenarios of immunoediting: elimination, equilibrium, and tumor escape. Ledzewicz and Moore \citep{Ledzewicz2016} analyzed a mathematical model for the treatment of chronic myeloid leukemia, which includes two types of leukemia cells and variable immune response strengths. They applied theoretical and numerical methods to explore the long-term dynamic of constant drug concentrations. In \citep{Fassoni2019}, a comprehensive mathematical model encompassing 20 different submodels was proposed to describe the interactions between immune effector cells and leukemic cells. The authors systematically compared the behaviors of each submodel based on data from TKI cessation studies. In addition, Fassoni et al. \citep{Fassoni2024} considered a specific model from \citep{Fassoni2019} and demonstrated the existence of limit cycles, which biologically align with the fluctuations of residual leukemic cell populations. Kim et al. \citep{Kim2008} combined experimental data with a delay differential model, reporting that the immune response to leukemia may help maintain remission in patients with chronic myeloid leukemia. Other studies have also addressed tumor-immune interactions in chronic myeloid leukemia \citep{Hahnel2020,Doumic-Jauffret2010,Ainseba2011,Besse2017,Karg2022,Paquin2011,Berezansky2015}.

Optimal control theory has been applied to tackle cancer treatment issues using tumor-immune models. Castiglione and Piccoli \citep{Castiglione2007} employed this theory to analyze the effect of immunotherapy. Meanwhile, de Pillis and Radunskaya \citep{dePillis2001} utilized optimal control theory to determine schedules for chemotherapy administration, resulting in total tumor size at the designated final time that is smaller than that achieved with traditional pulsed therapy. Building on the tumor–immune model introduced by Kirschner and Panetta \citep{Kirschner1998}, optimal control theory has been applied to optimize immunotherapy \citep{Burden2004,Ghaffari2010,SarvAhrabi2020}. To establish optimal therapeutic protocols, Burden et al. \citep{Burden2004} aimed to maximize the number of effector cells and interleukin-2, while simultaneously minimizing both the number of tumor cells and the cost of treatment. In later studies \citep{Ghaffari2010,SarvAhrabi2020}, a linear penalty was added to the objective functional compared to the one used in \citep{Burden2004}. Ghaffari and Naserifar \citep{Ghaffari2010} demonstrated that cancerous cells can be effectively eliminated in a shorter time frame than reported in \citep{Burden2004}. Additionally, Sarv Ahrabi and Momenzadeh \citep{SarvAhrabi2020} compared the outcomes obtained from three different numerical methods for solving optimal control problems. de Pillis et al. \citep{dePillis2007} explored a tumor-immune model that integrates chemotherapy and compared two distinct optimal control strategies: quadratic control and linear control. Their findings indicated that both quadratic and linear controls produced similar tumor evolution dynamics regarding chemical drug administration. Multiple studies have investigated different optimal control problems based on a model of immunological activity during cancer growth to identify the best treatment strategies \citep{Ledzewicz2012,Ledzewicz2013}. Utilizing the model from \citep{Ledzewicz2016}, Moore et al. \citep{Moore2018} applied control theory to optimize combined therapies and various approaches to determine superior treatment regimens. Furthermore, Ledzewicz and Moore \citep{Ledzewicz2018} found that both the original and extended models' optimal control problems \citep{Ledzewicz2016,Moore2018} yield qualitatively and quantitatively similar solutions under appropriate conditions. Additionally, A\"{i}nseba and Benosman \citep{Ainseba2010} explored therapeutic effects through five different therapeutic scenarios using the optimal control method to minimize the number of cancer cells and the cost of treatment. Their study revealed that therapies causing mortality in
cancer hematopoietic stem cells are effective. Based on the model proposed by \citep{Berezansky2015}, Bunimovich-Mendrazitsky and Shklyar \citep{Bunimovich-Mendrazitsky2017} investigated the optimal scheduling of combination treatments with imatinib and interferon-$\alpha$ using Pontryagin's minimum principle. Their research highlighted that the duration of imatinib administration is a critical factor for successful therapy.

Moore and Li \citep{Moore2004} proposed an ordinary differential equation model to investigate the interaction among naive T cells, effector T cells, and cancer cells. They reported that cancer cells' growth and death rates play critical roles in influencing cancer evolution.
The model has been further examined in many other studies \citep{Nanda2007,Dimitriu2019,Valle2021,Moore2018,Krishchenko2016}. Building on the established model, Nanda et al. \citep{Nanda2007} explored optimal combination treatment regimens involving targeted therapy and broad cytotoxic therapy under various assumptions. These results indicated that combination therapy is more effective than monotherapy. Dimitriu \citep{Dimitriu2019} conducted a global sensitivity analysis related to parameters in the leukemia model, revealing that the growth rate of chronic myeloid leukemia (CML) cells is the most sensitive parameter when CML cells are increasing. Valle et al. \citep{Valle2021} investigated personalized immunotherapy strategies by introducing a term representing the infused enhanced T lymphocytes. They formulated a therapeutic protocol aimed at eliminating the cancer cell population. Moore \citep{Moore20181} examined optimal control problems to optimize drug doses in a combination regimen, finding that the total levels of the two drug types decrease as the cancer cell death rate rises. Krishchenko and Starkov \citep{Krishchenko2016} provided upper bounds for three cell populations and determined the existence conditions of a positively invariant polytope. They also studied the globally asymptotic stability of tumor equilibrium due to the cyclic application of localizing functions. In 2018, Talkington et al. \citep{Talkington2018} modified Moore and Li's model \citep{Moore2004} to consider the logistic growth of tumor cells. They analyzed both the existence and locally asymptotic stability of tumor-free equilibrium, as well as the existence and locally asymptotic stability of positive equilibrium with a substantial number of tumor cells. Studying the interactions between tumor cells and T cells is crucial in oncology. Therefore, further research on the interesting phenomena of this model can provide valuable insights for tumor control.

In this study, we further investigate the tumor-immune model established by Talkington et al. \citep{Talkington2018} to assess the effect of immunotherapy. We conduct mathematical analysis and numerical simulations to explore the existence of equilibria and the global stability of the tumor-free equilibrium. A partial differential equation is established to describe the time required for the number of tumor cells to reach a specified level. Additionally, based on the proposed model, we examine optimal immunotherapeutic strategies. The model is presented in Section \ref{sec:2}. In Section \ref{sec:3}, we conduct theoretical investigations to analyze the positivity and boundaries of the solutions of the model equation. We also explore the existence of equilibria and the stability of the tumor-free equilibrium. Furthermore, we determine the first time point cancer cells reach a specified value and provide an upper bound for this time. Numerical analyses are performed to investigate the dynamics of the model. To optimally allocate the T cells administered during treatment, we propose and analyze an optimal control problem. We conduct numerical simulations to derive optimal immunotherapeutic protocols. Finally, we present our conclusions in Section \ref{sec:4}.

\section{Model formulation}
\label{sec:2}
This section introduces the model study in this paper.

An ordinary differential equation model has been established to describe the interaction among naive T cells, effector T cells, and CML cells \citep{Moore2004}.
Let $T$, $E$, and $N$ represent the populations of CML cells, effector, and naive T cells, respectively. The system of differential equations is given as:
\begin{equation}
\label{eq:model_0}
\left\{
\begin{aligned}
\dfrac{{\rm d} T}{{\rm d} t}& =  r_c T \ln \Big(  \dfrac{T_{\max} }{T} \Big) - d_c T - \gamma_c T E,\\
\dfrac{{\rm d} E}{{\rm d} t}& =  \alpha_n k_n N \Big(  \dfrac{T}{ T + g } \Big) + \alpha_e E \Big(  \dfrac{T}{ T + g } \Big) - d_e E - \gamma_e T E,   \\
\dfrac{{\rm d} N}{{\rm d} t}& =  s_n - d_n N - k_n N \Big(  \dfrac{T}{ T + \eta } \Big).
\end{aligned}
\right.
\end{equation}
In 2018, Talkington et al. \citep{Talkington2018} modified equation \eqref{eq:model_0} by replacing the Gompertzian growth of tumor cells with logistic growth, simplified the regulation of tumor cells in producing effector T cells, and obtained the following equation: 
\begin{equation}
\label{eq:model_1}
\left\{
\begin{aligned}
\dfrac{{\rm d} T}{{\rm d} t}& =  a T ( 1 - b T) - \gamma_c E T,\\
\dfrac{{\rm d} E}{{\rm d} t}& =  s_e - d_e E +
\alpha_n k_n N \Big(  \dfrac{T}{ T + g } \Big)  - \gamma_e E T,   \\
\dfrac{{\rm d} N}{{\rm d} t}& =  s_n - d_n N - k_n N \Big(  \dfrac{T}{ T + g } \Big).
\end{aligned}
\right.
\end{equation} 
Here, $a$ denotes the tumor growth rate; $b$ represents the inverse of the tumor carrying capacity; $\gamma_c$ is the rate of tumor cell loss due to encounters with effector T cells; $s_e$ is the rate at which effector T cells enter the bloodstream; $d_e$ is the death rate of effector T cells; $k_n$ is the activation rate of naive T cells, such that the Michaelis–Menten term $ k_n N \Big(  \dfrac{T}{ T + g } \Big)$ represents the instantaneous loss rate of naive T cells when they contact myeloid leukemia antigen; $\alpha_n$ is the activation coefficient of naive T cells encountering an antigen-presenting cell (APC), which then proliferate and divide into effector T cells; $g$ is the half-saturation coefficient; $\gamma_e$ is the rate constant for the loss of effector T cells due to encounters with myeloid leukemia cells; $s_n$ is the rate at which naive T cells enter the bloodstream; and $d_n$ is the death rate constant of naive T cells. All parameters are assumed to be positive unless otherwise specified. Given the biological context, we choose nonnegative initial conditions. Therefore,  based on the standard theory of ordinary differential equations, equation \eqref{eq:model_1} yields a unique solution for any nonnegative initial conditions. 

\section{Results}
\label{sec:3}

Equation \eqref{eq:model_1} presents the model in this study. We begin by performing a mathematical analysis of the dynamics of the model system \eqref{eq:model_1}. Next, we discuss the biological significance of the analytical and numerical simulation results. Finally, we consider the optimal immunotherapeutic protocols.

\subsection{Positivity and boundary of solutions}
In this subsection, we show that the solutions of equation \eqref{eq:model_1} are positive and bounded over time when the initial conditions satisfy $T(0) > 0$, $E(0) \geq 0$ and $N(0) \geq 0$.

\begin{theorem}
	\label{thm2.1}
	The solutions of equation \eqref{eq:model_1} are positive for $t>0$ when the initial conditions satisfy $T(0) > 0$, $E(0) \geq 0$ and $N(0) \geq 0$.
\end{theorem}

\begin{proof} Let $(T(t), E(t), N(t))$ denote the solution of equation \eqref{eq:model_1} with initial condition $(T(0), E(0), N(0))$.
	
From the first equation of \eqref{eq:model_1}, we have
\begin{equation}
\label{eq:cancer_1}
T(t) = T(0) {\mathrm e}^{ \int_{ 0}^{t} \big( a ( 1- bT(s)) - \gamma_c E(s)  \big) {\rm d} s},
\end{equation}
which implies that $T(t) > 0$ for $t \geq 0$. By the method of variation of constant, the third equation of \eqref{eq:model_1} yields
\begin{equation*}
N(t) = {\mathrm e}^{ - \int_{ 0}^{t} h(s) {\rm d} s} \left( N(0)  + s_n  \int_{0}^{t} {\mathrm e}^{ \int_{ 0}^{s_1} h(s) {\rm d} s} {\rm d} s_1 \right),\ \mathrm{where}\ h(s) = d_n + k_n \dfrac{T(s)}{ T(s) + g}.
\end{equation*}
Thus, it is easy to have $N(t) > 0$ for all $t > 0$. 
	
Finally, based on the above arguments, the second equation of \eqref{eq:model_1} implies
\begin{equation*}
\dfrac{ {\rm d} E}{ {\rm d} t} \geq s_e - d_e E - \gamma_e E T.
\end{equation*}
Thus, 
\begin{equation*}
E(t) \geq {\mathrm e}^{ - \int_{ 0}^{t} \big( d_e + \gamma_e T(s) \big) {\rm d} s} \left( E(0)  + s_e  \int_{0}^{t} {\mathrm e}^{ \int_{ 0}^{s_1} \big( d_e + \gamma_e T(s)   \big) {\rm d} s} {\rm d} s_1 \right).
\end{equation*}
Hence, $E(t) > 0$ for $t > 0$. This completes the proof.
\end{proof}

\begin{remark}
\label{rem:1}
Similar to the proof of Theorem \ref{thm2.1}, we can conclude that the solutions of equation \eqref{eq:model_1} satisfy $T(t) = 0$, $E(t) > 0$ and $N(t) > 0$ for $t >0$ when the initial conditions $T(0) = 0$, $E(0) \geq 0$ and $N(0) \geq 0$. 
\end{remark}

Now, we consider the boundary of the solutions.

\begin{theorem}
\label{thm2.1_b}
The solutions of equation \eqref{eq:model_1} have the ultimate upper boundary in $\mathbb{R}^3$ when the initial conditions satisfy $T(0) \geq 0$, $E(0) \geq 0$ and $N(0) \geq 0$.
\end{theorem}

\begin{proof}
From the first equation of \eqref{eq:model_1}, we have $\dfrac{{\rm d} T}{{\rm d} t} \leq a T(1- b T)$, which implies
\begin{equation}
\label{eq:T_ub}
\limsup_{\substack t \rightarrow \infty} T(t) \leq \dfrac{1}{b}.
\end{equation}
The third equation of \eqref{eq:model_1} implies $\dfrac{{\rm d} N}{{\rm d} t} \leq s_n - d_n N$. Hence, 
\begin{equation*}
\limsup_{\substack t \rightarrow \infty} N(t) \leq \dfrac{s_n}{ d_n}.
\end{equation*}
Thus, for any $\epsilon > 0$, there exists a finite time $t_1$ such that the solutions of system \eqref{eq:model_1} satisfy $ T \leq \dfrac{1}{b} + \epsilon$ and $N(t) \leq \dfrac{s_n}{ d_n} + \epsilon$ for $t > t_1$. 
	
From the second equation of \eqref{eq:model_1}, we have $\dfrac{{\rm d} E}{{\rm d} t} \leq s_e - d_e E + \alpha_n k_n (\dfrac{s_n}{ d_n} + \epsilon) \dfrac{ 1/b + \epsilon}{ 1/b + \epsilon + g}$ when $t > t_1$, which gives
\begin{equation*}
\limsup_{\substack t \rightarrow \infty} E(t) \leq \dfrac{1}{ d_e} \Big( s_e + \alpha_n k_n (\dfrac{s_n}{ d_n} + \epsilon) \dfrac{ 1 + b \epsilon}{ 1 + b \epsilon + b g} \Big).
\end{equation*}
Since $\epsilon$ is arbitrarily small, we obtain 
\begin{equation*}
\limsup_{\substack t \rightarrow \infty} E(t) \leq \dfrac{1}{ d_e} \Big( s_e + \dfrac{\alpha_n k_n s_n }{ d_n (1 +  b g) } \Big)
\end{equation*} 
by letting $\epsilon$ approach to $0$. This completes the proof.
\end{proof}

To determine the asymptotic behavior of cell populations around an equilibrium $S^* = (T^*, E^*, N^*)$, we linearize the equation \eqref{eq:model_1} near the equilibrium and obtain the Jacobian matrix
\begin{equation}
\label{eq:j1}
J = \left.\left( 
\begin{array}{ccc}
a - 2 abT - \gamma_c E    & - \gamma_c T  & 0 \\
\dfrac{g \alpha_n k_n N }{ (T + g)^2 } -  \gamma_e E  &  -d_e -  \gamma_e T  & \dfrac{ \alpha_n k_n T }{ T + g } \\
-  \dfrac{g k_n N }{ (T + g)^2 }  & 0  & -d_n - \dfrac{ k_n T }{ T + g }
\end{array}
\right)\right\vert_{(T^*, E^*, N^*)}.
\end{equation}
The equilibrium $S^*$ is asymptotically stable when all eigenvalues of the Jacobian matrix $J$ have negative real parts.

It is easy to have the tumor-free equilibrium $S_1^* = (0, \dfrac{s_e}{d_e}, \dfrac{s_n}{d_n})$ for  equation \eqref{eq:model_1}. Thus, according to Theorem 3 of \cite{Talkington2018} and the Jacobian matrix \eqref{eq:j1}, we have the conclusion below.
\begin{theorem}	
\label{thm2.a}
The tumor-free equilibrium $S_1^*$ is locally asymptotically stable when $a < \dfrac{ \gamma_c s_e}{d_e}$.
\end{theorem}

The following theorem gives the existence of positive equilibria.
\begin{theorem}
\label{thm2.b}
Consider the system \eqref{eq:model_1}, if $a > \dfrac{ \gamma_c s_e}{d_e}$, there exists at least one positive steady state $S_2^* = (T^*, E^*, N^*)$. Furthermore, denote 
\begin{eqnarray*}
\beta_1 &=& \dfrac{d_n g s_e}{s_e (d_n +k_n)+ \alpha_n k_n s_n },\\
\beta_2 &=& \dfrac{d_e(d_n +k_n) + \gamma_e d_n g}{\gamma_e (d_n +k_n) },\\
\beta_3 &=& \dfrac{d_e d_n g}{ \gamma_e (d_n +k_n) },\\
\beta_4 &=&  \frac{\gamma_e (d_n +k_n) }{s_e (d_n +k_n)+ \alpha_n k_n s_n},
\end{eqnarray*}
and let
$$\Delta^* = \max_{0\leq T^*\leq 1/b} \dfrac{ (T^*+ \beta_1)^2 - ( \beta_1^2 + \beta_3 - \beta_1 \beta_2)}{ ({T^*}^2 + \beta_2 T^* + \beta_3)^2},$$	 
the system \eqref{eq:model_1} has a unique positive equilibrium if 
\begin{equation}
\label{eq:psd}
\Delta^* \leq \beta_4 \frac{a b}{\gamma_c}.
\end{equation}
\end{theorem}

\begin{proof}
The positive steady states of equation \eqref{eq:model_1} are given by solutions of the following algebraic equation:
\begin{equation}
\label{eq:ps_1}
\left\{
\begin{array}{rcl}
a ( 1 - b T^*) - \gamma_c E^* & = & 0,\\
s_e - d_e E^* + \alpha_n k_n N^* \Big(  \dfrac{T^*}{ T^* + g } \Big)  - \gamma_e E^* T^* & = & 0,   \\
s_n - d_n N^* - k_n N^* \Big(  \dfrac{T^*}{ T^* + g } \Big) & = & 0.
\end{array}
\right.
\end{equation} 
From the third equation of \eqref{eq:ps_1}, we have 
\begin{equation}
\label{eq:ps_2}
N^* = \dfrac{s_n}{ d_n + k_n T^*/(T^*+g) }  .
\end{equation} 
Substituting \eqref{eq:ps_2} into the second equation of \eqref{eq:ps_1}, we obtain
\begin{equation}
\label{eq:ps_3}
E^* = \dfrac{1}{d_e + \gamma_e T^*} \left( s_e + \dfrac{ \alpha_n k_n s_n T^*/(T^*+g)}{ d_n + k_n T^* /(T^*+g) } \right)  \triangleq f_1(T^*).
\end{equation} 
Moreover, the first equation of \eqref{eq:ps_1} gives
\begin{equation}
\label{eq:ps_4}
E^* = a (1-bT^*)/\gamma_c \triangleq f_2(T^*).
\end{equation} 
Set 
\begin{equation}
\label{eq:ps_4aa}
f(T^*)= f_1(T^*) - f_2(T^*). 
\end{equation} 

Next, we need to show that the equation $f(T^*)=0$ has at least one root in the interval $(0, 1/b)$, which implies that equation \eqref{eq:model_1} has at least one positive steady state. 
	
It is easy to verify that when $a > \dfrac{\gamma_c s_e}{d_e}$, $f(0) = \dfrac{s_e}{d_e} - \dfrac{a}{\gamma_c} < 0$ and $f( \dfrac{1}{b}) >0$. Moreover, the function $f$ is continuous in the interval $[0, 1/b]$. Hence, there exists at least one $T^* \in (0, 1/b)$ that satisfies $f(T^*) =0$.
	
We have proven the existence of positive equilibria. Next, to establish the uniqueness of steady states, i.e., the uniqueness of the roots of $f$ on $[0, 1/b]$, we only need to show that $f'(T^*) \geq 0$ for and $T^*\in (0, 1/b)$.
				
From \eqref{eq:ps_3}, we have 
 \begin{equation}
\label{eq:re_ap3}
\begin{aligned}
f_1(T^*) &= \dfrac{1}{d_e + \gamma_e T^*} \left( s_e + \dfrac{ \alpha_n k_n s_n T^*/(T^*+g)}{ d_n + k_n T^* /(T^*+g) } \right)  \\
&= \dfrac{1}{d_e + \gamma_e T^*} \left( s_e + \dfrac{ \alpha_n k_n s_n T^*}{ d_n (T^*+g) + k_n T^* } \right) \\
&= \dfrac{\big( s_e (d_n +k_n)+ \alpha_n k_n s_n \big) T^* + d_n g s_e}{ \gamma_e (d_n +k_n) {T^*}^2 + \big( d_e(d_n +k_n) + \gamma_e d_n g\big)T^* + d_e d_n g} \\
&= \dfrac{1}{\beta_4} \dfrac{T^*+ \beta_1}{ {T^*}^2 + \beta_2 T^* + \beta_3}. 
\end{aligned}
\end{equation}
Thus, from the condition \eqref{eq:psd}, we have
\begin{equation}
\label{eq:re_ap4}
f_1'(T^*) = - \frac{1}{\beta_4}\dfrac{ (T^*+ \beta_1)^2 - ( \beta_1^2 + \beta_3 - \beta_1 \beta_2)}{ ({T^*}^2 + \beta_2 T^* + \beta_3)^2}\geq -\frac{ab}{\gamma_c} = f_2'(T^*) 
\end{equation}
for any $T^* \in (0, 1/b)$. Hence, we always have $f'(T^*) > 0$ for $T^*\in (0, 1/b)$, the equation $f(T^*) =0$ has a unique solution in the interval $(0, 1/b)$.
\end{proof}

Now, we explore the asymptotic stability of positive steady states.
\begin{theorem}
\label{thm2.2_p}
Consider the system \eqref{eq:model_1}, assuming that $a > \dfrac{ \gamma_c s_e}{d_e}$, and
\begin{equation}
\label{eq:re_ap11}
\begin{aligned}
 \gamma_e < \min \left\{ b d_e, \dfrac{d_e g \alpha_n k_n d_n s_n b^2}{(s_e k_n + \alpha_n k_n s_n) \big( d_n (1+ b g)+ k_n  \big)} \right\},
\end{aligned}
\end{equation}
the positive steady state $S_1^*$ is locally asymptotically stable.
\end{theorem}
\begin{proof} According to Theorem \ref{thm2.b}, there exists at least one positive steady state of the system \eqref{eq:model_1}. Let $(T^*, E^*, N^*)$ be the positive steady state of \eqref{eq:model_1}. We linearize equation \eqref{eq:model_1} at $(T^*, E^*, N^*)$ and obtain the following characteristic equation
\begin{equation}
\label{eq:re_ap12}
\lambda^3 + c_1 \lambda^2 + c_2 \lambda +c_3=0,
\end{equation}
where 
\begin{equation}
\label{eq:re_ap13}
\begin{aligned}
c_1 &= ab T^* + d_e + \gamma_e T^* + d_n + \dfrac{k_n T^*}{ T^* +g}, \\
c_2 &= ab T^*(d_e + \gamma_e T^*) + \Big(d_n + \dfrac{k_n T^*}{ T^* +g} \Big) (d_e + \gamma_e T^*)  + ab T^* \Big(d_n + \dfrac{k_n T^*}{ T^* +g} \Big) + \gamma_c T^* \left( \dfrac{g \alpha_n k_n N^* }{(T^*+g)^2} - \gamma_e E^* \right), \\
c_3 & = T^* \left( ab(d_e + \gamma_e T^*) \Big(d_n + \dfrac{k_n T^*}{ T^* +g} \Big) + \gamma_c \dfrac{g \alpha_n k_n N^* }{(T^*+g)^2} d_n  - \gamma_c \gamma_e E^* \Big(d_n + \dfrac{k_n T^*}{ T^* +g} \Big)  \right).
\end{aligned}
\end{equation}

Based on the Routh-Hurwitz stability criterion, all roots of \eqref{eq:re_ap12} have negative real parts if and only if
$c_1 > 0$, $c_1 c_2 - c_3 > 0$ and $c_3 > 0$. It is easy to have $c_1 > 0$. Hence, we only need to find the conditions so that $c_1 c_2 - c_3 > 0$ and $c_3 > 0$. 

From the conditions \eqref{eq:ps_2}, \eqref{eq:ps_3}, and $a > \dfrac{ \gamma_c s_e}{d_e}$, we have 
 \begin{eqnarray*}
c_3 &= & T^* \left( ab(d_e + \gamma_e T^*) (d_n + \dfrac{k_n T^*}{ T^* +g}) + \gamma_c \dfrac{g \alpha_n k_n N^* }{(T^*+g)^2} d_n  - \gamma_c \gamma_e E^* (d_n + \dfrac{k_n T^*}{ T^* +g})  \right) \\
&> & T^* \left( \dfrac{\gamma_c s_e}{d_e} b(d_e + \gamma_e T^*) (d_n + \dfrac{k_n T^*}{ T^* +g}) + \gamma_c \dfrac{g \alpha_n k_n N^* }{(T^*+g)^2} d_n  - \gamma_c \gamma_e E^* (d_n + \dfrac{k_n T^*}{ T^* +g}) \right)\\
&> & T^* \left( \dfrac{\gamma_c s_e}{d_e} b d_e d_n + \gamma_c \dfrac{g \alpha_n k_n N^* }{(T^*+g)^2} d_n  - \gamma_c \gamma_e E^* (d_n + \dfrac{k_n T^*}{ T^* +g}) \right)\\
&= & T^* \left( \dfrac{\gamma_c s_e}{d_e} b d_e d_n + \dfrac{g \alpha_n k_n d_n}{(T^*+g)^2} \dfrac{ \gamma_c s_n}{ d_n + k_n T^*/(T^*+g) }   -  (d_n + \dfrac{k_n T^*}{ T^* +g}) \dfrac{\gamma_c \gamma_e}{d_e + \gamma_e T^*} \left( s_e + \dfrac{ \alpha_n k_n s_n T^*/(T^*+g)}{ d_n + k_n T^* /(T^*+g) } \right) \right)  \\
&= &T^* \left( \dfrac{\gamma_c s_e}{d_e} b d_e d_n + \dfrac{ g \alpha_n k_n d_n\gamma_c s_n}{ d_n (T^*+g)^2+ k_n T^* (T^*+g) } -  \dfrac{\gamma_c \gamma_e}{d_e + \gamma_e T^*} \big( s_e d_n + (s_e k_n + \alpha_n k_n s_n) T^*/(T^*+g) \big) \right)  \\
&\geq & T^* \left( \gamma_c s_e b d_n + \dfrac{ g \alpha_n k_n d_n\gamma_c s_n}{ d_n (1/b+g)^2+ k_n/b (1/b+g) } -  \dfrac{\gamma_c \gamma_e}{d_e } \big( s_e d_n + (s_e k_n + \alpha_n k_n s_n) T^*/(T^*+g) \big) \right) \\
&= & T^* \left( \gamma_c s_e b d_n + \dfrac{ g \alpha_n k_n d_n\gamma_c s_n b^2}{ d_n (1+ b g)^2+ k_n (1+bg) } -  \dfrac{\gamma_c \gamma_e}{d_e } s_e d_n -\dfrac{\gamma_c \gamma_e}{d_e } (s_e k_n + \alpha_n k_n s_n) T^*/(T^*+g) \right)  \\
&\geq &T^* \left( \gamma_c s_e b d_n + \dfrac{ g \alpha_n k_n d_n\gamma_c s_n b^2}{ d_n (1+ b g)^2+ k_n (1+bg) } -  \dfrac{\gamma_c \gamma_e}{d_e } s_e d_n -\dfrac{\gamma_c \gamma_e}{d_e } (s_e k_n + \alpha_n k_n s_n) /(1+ b g)  \right) \\
&=& T^* \left( \gamma_c s_e d_n \dfrac{ b d_e - \gamma_e}{d_e} + \gamma_c \dfrac{ d_e g \alpha_n k_n d_n s_n b^2 - \gamma_e (s_e k_n + \alpha_n k_n s_n) \big( d_n (1+ b g)+ k_n  \big)}{ d_e d_n (1+ b g)^2+d_e k_n (1+bg) } \right) \\
&>&0.
\end{eqnarray*}

We have 
 \begin{eqnarray*}
 c_1 c_2 - c_3  
&= & ab T^* \left( ab T^*(d_e + \gamma_e T^*) + ab T^* \Big(d_n + \dfrac{k_n T^*}{ T^* +g} \Big) + \gamma_c T^* \dfrac{g \alpha_n k_n N^* }{(T^*+g)^2}  \right) - ab T^* \gamma_c T^* \gamma_e E^* \\
&&{}+ (d_e + \gamma_e T^*) \left( ab T^*(d_e + \gamma_e T^*) + \Big(d_n + \dfrac{k_n T^*}{ T^* +g} \Big) (d_e + \gamma_e T^*)  + ab T^* \Big(d_n + \dfrac{k_n T^*}{ T^* +g} \Big) + \gamma_c T^*  \dfrac{g \alpha_n k_n N^* }{(T^*+g)^2}  \right) \\
&&{} - (d_e + \gamma_e T^*) \gamma_c T^* \gamma_e E^*  \\
&&{}+ \Big(d_n + \dfrac{k_n T^*}{ T^* +g} \Big) \left( ab T^*(d_e + \gamma_e T^*) +  \Big(d_n + \dfrac{k_n T^*}{ T^* +g} \Big) (d_e + \gamma_e T^*)  + ab T^* \Big(d_n + \dfrac{k_n T^*}{ T^* +g} \Big)  \right)  \\
&&{} + \gamma_c T^* \dfrac{g \alpha_n k_n N^* }{(T^*+g)^2} \dfrac{k_n T^*}{ T^* +g} 
\end{eqnarray*}
 \begin{eqnarray*}
~~~~~~~~~~~&> & ab T^* \left( ab T^*(d_e + \gamma_e T^*) + ab T^* \Big(d_n + \dfrac{k_n T^*}{ T^* +g} \Big) + \gamma_c T^* \dfrac{g \alpha_n k_n N^* }{(T^*+g)^2}  \right) - ab T^* \gamma_c T^* \gamma_e E^* \\
&&{}+ (d_e + \gamma_e T^*) \left( ab T^*(d_e + \gamma_e T^*)  + ab T^* \Big(d_n + \dfrac{k_n T^*}{ T^* +g} \Big) + \gamma_c T^*  \dfrac{g \alpha_n k_n N^* }{(T^*+g)^2}  \right) - (d_e + \gamma_e T^*) \gamma_c T^* \gamma_e E^* \\
&= & (ab T^* +d_e + \gamma_e T^* ) T^* \left( ab (d_e + \gamma_e T^*) + ab  \Big(d_n + \dfrac{k_n T^*}{ T^* +g} \Big) + \gamma_c  \dfrac{g \alpha_n k_n N^* }{(T^*+g)^2} - \gamma_c  \gamma_e E^* \right)
\end{eqnarray*}

Thus, let
$$A =  ab (d_e + \gamma_e T^*) + ab  \Big(d_n + \dfrac{k_n T^*}{ T^* +g} \Big) + \gamma_c  \dfrac{g \alpha_n k_n N^* }{(T^*+g)^2} - \gamma_c  \gamma_e E^*,
$$
from \eqref{eq:ps_2}, \eqref{eq:ps_3}, $a > \dfrac{ \gamma_c s_e}{d_e}$, and \eqref{eq:re_ap11}, we obtain 
 \begin{eqnarray*}
A &=& ab(d_e + \gamma_e T^*) + ab \Big(d_n + \dfrac{k_n T^*}{ T^* +g} \Big) + \gamma_c \dfrac{g \alpha_n k_n N^* }{(T^*+g)^2}   - \gamma_c \gamma_e E^* \\
&> & \dfrac{\gamma_c s_e}{d_e} b(d_e + \gamma_e T^*) + \dfrac{\gamma_c s_e}{d_e} b \Big(d_n + \dfrac{k_n T^*}{ T^* +g} \Big) + \gamma_c \dfrac{g \alpha_n k_n N^* }{(T^*+g)^2}   - \gamma_c \gamma_e E^*  \\
&> & \dfrac{\gamma_c s_e}{d_e} b d_e + \dfrac{\gamma_c s_e}{d_e} b d_n + \gamma_c \dfrac{g \alpha_n k_n N^* }{(T^*+g)^2}   - \gamma_c \gamma_e E^*  \\
&=&   \dfrac{\gamma_c s_e}{d_e} b ( d_e + d_n) + \dfrac{g \alpha_n k_n }{(T^*+g)^2} \dfrac{ \gamma_c s_n}{ d_n + k_n T^*/(T^*+g) }   - \dfrac{\gamma_c \gamma_e}{d_e + \gamma_e T^*} \left( s_e + \dfrac{ \alpha_n k_n s_n T^*/(T^*+g)}{ d_n + k_n T^* /(T^*+g) } \right)   \\
&= & \dfrac{\gamma_c s_e}{d_e} b (d_e+ d_n) + \dfrac{ g \alpha_n k_n \gamma_c s_n}{ d_n (T^*+g)^2+ k_n T^* (T^*+g) } -  \dfrac{\gamma_c \gamma_e}{d_e + \gamma_e T^*} \left( s_e + \dfrac{ \alpha_n k_n s_n T^*/(T^*+g)}{ d_n + k_n T^* /(T^*+g) } \right)   \\
&\geq &  \dfrac{\gamma_c s_e}{d_e} b (d_e+ d_n) + \dfrac{ g \alpha_n k_n \gamma_c s_n}{ d_n (T^*+g)^2+ k_n T^* (T^*+g) } -  \dfrac{\gamma_c \gamma_e}{d_e } \left( s_e + \dfrac{ \alpha_n k_n s_n T^*}{ (d_n + k_n) T^* +d_n g } \right)    \\
&\geq &  \dfrac{\gamma_c s_e}{d_e} b (d_e+ d_n) + \dfrac{ g \alpha_n k_n \gamma_c s_n}{ d_n (1/b+g)^2+ k_n  (1/b+g) /b } -  \dfrac{\gamma_c \gamma_e}{d_e } \left( s_e + \dfrac{ \alpha_n k_n s_n /b}{ (d_n + k_n) /b +d_n g } \right)    \\
&= &  \dfrac{\gamma_c s_e}{d_e} b (d_e+ d_n) + \dfrac{ g \alpha_n k_n \gamma_c s_n b^2}{ d_n (1+bg)^2+ k_n  (1+ bg)  } -  \dfrac{\gamma_c \gamma_e}{d_e } \left( s_e + \dfrac{ \alpha_n k_n s_n }{ (d_n + k_n)  +b d_n g } \right)    \\
&= &  \dfrac{\gamma_c s_e}{d_e} b (d_e+ d_n) + \dfrac{ g \alpha_n k_n \gamma_c s_n b^2}{ d_n (1+bg)^2+ k_n  (1+ bg)  } -  \dfrac{\gamma_c \gamma_e}{d_e } s_e  - \dfrac{\gamma_c \gamma_e}{d_e }  \dfrac{ \alpha_n k_n s_n }{ (d_n + k_n)  +b d_n g }    \\
&=&  \dfrac{\gamma_c s_e}{d_e} \big( b (d_e+ d_n) -\gamma_e  \big)   +    \dfrac{ g \alpha_n k_n \gamma_c s_n b^2d_e - \gamma_c \gamma_e \alpha_n k_n s_n (1+bg)}{ d_e d_n (1+bg)^2+ d_e k_n  (1+ bg)  }     \\
&=&  \dfrac{\gamma_c s_e}{d_e} \big( b (d_e+ d_n) -\gamma_e  \big)   + \alpha_n k_n \gamma_c s_n  \dfrac{ g  b^2d_e - \gamma_e  (1+bg)}{ d_e d_n (1+bg)^2+ d_e k_n  (1+ bg)  } 
\end{eqnarray*}

Moreover, we have 
 \begin{eqnarray*}
g b^2 d_e - \gamma_e (1+bg) &>&g b^2 d_e - 
\dfrac{d_e g \alpha_n k_n d_n s_n b^2}{(s_e k_n + \alpha_n k_n s_n) \big( d_n (1+ b g)+ k_n  \big)}(1+bg) \\
&>& g b^2 d_e - 
\dfrac{d_e g \alpha_n k_n d_n s_n b^2}{ \alpha_n k_n s_n  d_n (1+ b g)}(1+bg) \\
&=&0.
\end{eqnarray*}
Thus, $A > 0$, and hence $c_1 c_2 - c_3 > 0$.

Thus, from the Routh-Hurwitz stability criterion, the positive steady state of \eqref{eq:model_1} is asymptotically stable when the condition \eqref{eq:re_ap11} is satisfied.
\end{proof}

Next, we investigate the global stability of the tumor-free equilibrium $S_1^*$ in $\Omega$, where $\Omega$ is given as
\begin{equation}
\label{eq:bd_region}
\Omega = \{ (T, E, N) \in \mathbb{R} ^3 \mid T >0,   E \geq 0,  N \geq 0 \}. 
\end{equation}
From Theorem \ref{thm2.1} and Theorem \ref{thm2.1_b}, $\Omega$ is an invariant set for equation \eqref{eq:model_1}.

\begin{theorem}
\label{thm2.2}
Consider the model equation \eqref{eq:model_1}, if $a < \dfrac{\gamma_c s_e} {d_e + \gamma_e/b}$, the tumor-free equilibrium $S_1^*$ is globally asymptotically stable in $\Omega$.
\end{theorem}

\begin{proof} 
From Theorem \ref{thm2.1} and equation \eqref{eq:model_1}, we have
\begin{equation*}
\dfrac{ {\rm d} E}{ {\rm d} t} \geq s_e - d_e E - \dfrac{\gamma_e}{b} E .
\end{equation*}
Applying the comparison principle,
\begin{equation}
\label{eq:E0_low}
E(t) \geq \dfrac{s_e}{d_e + \gamma_e/b} + \Big( E(0) - \dfrac{s_e}{d_e + \gamma_e/b} \Big) {\mathrm e}^{ - (d_e + \gamma_e/b) t }.
\end{equation}
Thus, for any $\epsilon > 0$, there exists a time $t_2$ such that
\begin{equation*}
E(t) \geq \dfrac{s_e}{d_e + \gamma_e/b} - \epsilon
\end{equation*}
for $t \geq t_2$. 
	
From the first equation of \eqref{eq:model_1}, we have
\begin{equation}
\label{eq:thm_1}
\dfrac{ {\rm d} T}{ {\rm d} t} \leq ( a- \gamma_c E) T \leq ( a- \dfrac{ \gamma_c s_e}{d_e + \gamma_e/b} + \epsilon) T.
\end{equation}
Thus, we choose $\epsilon = \dfrac{1}{2} \Big( \dfrac{ \gamma_c s_e}{d_e + \gamma_e/b} - a \Big) > 0$, the equation \eqref{eq:thm_1} yields
\begin{equation}
\label{eq:thm_2}
\dfrac{ {\rm d} T}{ {\rm d} t}  \leq \dfrac{1}{2} ( a- \dfrac{ \gamma_c s_e}{d_e + \gamma_e/b} ) T.
\end{equation}
Since $a < \dfrac{ \gamma_c s_e}{d_e + \gamma_e/b}$, it follows from \eqref{eq:thm_2} that $ T(t) \rightarrow 0$ when $t \rightarrow \infty$. Therefore, the second and the third equations together imply $ E \rightarrow \dfrac{s_e}{d_e}$ and $ N \rightarrow \dfrac{s_n}{d_n}$ when $t \rightarrow \infty$. The above arguments imply that the tumor-free equilibrium $S_1^*$ is globally attractive in $\Omega$ when $a < \dfrac{ \gamma_c s_e}{d_e + \gamma_e/b}$.
	
Finally, from Theorem \ref{thm2.a}, $S_1^*$ is locally stable when $a < \dfrac{ \gamma_c s_e}{d_e + \gamma_e/b}$. Thus, $S_1^*$ is globally asymptotically stable. The proof is complete.
\end{proof}

Next, we examined the long-term behavior of tumor cells with varying tumor growth rates. According to Theorem \ref{thm2.a} and Theorem \ref{thm2.2}, the tumor-free equilibrium is locally asymptotically stable when $a < \dfrac{ \gamma_c s_e}{d_e}$ and globally asymptotically stable in the region $\Omega$ when $a < \dfrac{\gamma_c s_e} {d_e + \gamma_e/b}$. Furthermore, Theorem \ref{thm2.b} indicates that a positive steady state exists when $a > \dfrac{ \gamma_c s_e}{d_e}$. We set parameters $(s_e, s_n) = (1.5, 4)$. Based on the parameter values presented in Table \ref{tab:table1}, we have $\dfrac{ \gamma_c s_e}{d_e} = 0.125$ and $\dfrac{\gamma_c s_e} {d_e + \gamma_e/b} \approx 0.00245$. We selected six different values of $a$ and initialized the conditions with $(T(0), E(0), N(0) =(10000, 20, 100)$. The time courses of the cancer cell concentration are plotted in Fig. \ref{fig:fig0}. The cancer cells are eliminated when $a < \dfrac{\gamma_c s_e} {d_e + \gamma_e/b}$, while a coexistence state of immune cells and cancer cells can occur if $a > \dfrac{ \gamma_c s_e}{d_e}$. As illustrate in Fig. \ref{fig:fig0}, when $\dfrac{\gamma_c s_e} {d_e + \gamma_e/b} <a < \dfrac{ \gamma_c s_e}{d_e}$, the cancer cells may either be eliminated (for $a = 0.03$ and $0.06$) or coexist with immune cells (for $a = 0.09$) depending on the selected parameter values.

\begin{figure}[htbp]
	\centering
	\includegraphics[width=9 cm]{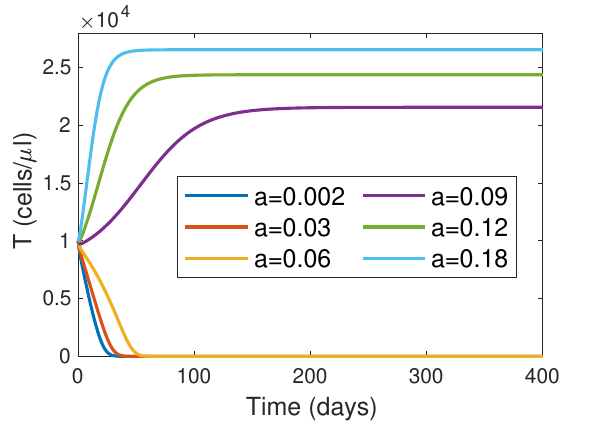}
	\caption{ \textbf{Time courses of the cancer cell concentration $T(t)$ with different values of $a$.}  Here, we set $(s_e, s_n)=(1.5, 4)$, and the other parameter values are the same as those in Table \ref{tab:table1}. }
	\label{fig:fig0}
\end{figure}

Similar to Theorems \ref{thm2.a}, \ref{thm2.b}, and \ref{thm2.2}, we explore the dynamics of the model equation \eqref{eq:model_1} when there are no effector T cells $(s_e = 0)$ or no naive T cells $(s_n = 0)$. 

\begin{theorem}
\label{thm2.aaa}
Considering the equation \eqref{eq:model_1}, we have the following conclusions:
\begin{enumerate}
\item[\rm{(i)}] If $s_e = s_n = 0$, the equation \eqref{eq:model_1} has two equilibria $(0, 0, 0)$ and $(\dfrac{1}{b}, 0, 0)$. Furthermore, the steady state $(0, 0, 0)$ is a locally unstable saddle point, while the state $(\dfrac{1}{b}, 0, 0)$ is globally asymptotically stable within the domain $\Omega$.
\item[\rm{(ii)}] If $s_e =0$ and $s_n > 0$, the equation \eqref{eq:model_1} has a locally unstable tumor-free equilibrium at $(0, 0, \dfrac{s_n}{d_n})$, along with at least one positive equilibrium.
\item[\rm{(iii)}] If $s_e >0$ and $s_n = 0$, the system \eqref{eq:model_1} has a tumor-free equilibrium at $(0, \dfrac{s_e}{d_e}, 0)$. If $a>\dfrac{\gamma_c s_e}{d_e}$, there exists at least one equilibrium of type $(T^*, E^*, 0)$ where $T^*, E^* > 0$. Conversely, if $a < \dfrac{\gamma_c s_e}{d_e}$, the tumor-free equilibrium is locally asymptotically stable, whereas if $a > \dfrac{\gamma_c s_e}{d_e}$, the tumor-free equilibrium becmes unstable.
\end{enumerate}
\end{theorem}

\begin{proof} We only give the proof for the globally asymptotically stable equilibrium $(\dfrac{1}{b}, 0, 0)$ in part \rm{(i)}. The other parts can be proved analogously to Theorems \ref{thm2.a}, \ref{thm2.b}, and \ref{thm2.2}. 
	
To prove (i), it is sufficient to verify the global attractiveness and the local stability of the equilibrium $(\dfrac{1}{b}, 0, 0)$.
	
When $s_e = s_n = 0$, equation \eqref{eq:model_1} can be written as
\begin{equation}
\label{eq:model_1_n}
\left\{
\begin{aligned}
\dfrac{{\rm d} T}{{\rm d} t}& =  a T ( 1 - b T) - \gamma_c E T,\\
\dfrac{{\rm d} E}{{\rm d} t}& =  - d_e E + \alpha_n k_n N \Big(  \dfrac{T}{ T + g } \Big)  - \gamma_e E T,   \\
\dfrac{{\rm d} N}{{\rm d} t}& =  - d_n N - k_n N \Big(  \dfrac{T}{ T + g } \Big).
\end{aligned}
\right.
\end{equation} 
The Jacobian matrix around the equilibrium $(\dfrac{1}{b}, 0, 0)$ is given by
\begin{equation}
\label{eq:j2}
J_1 = \left( 
\begin{array}{ccc}
-a    & - \dfrac{\gamma_c}{b}   & 0 \\
0  &  -d_e - \dfrac{\gamma_e}{b} ~~ & \dfrac{ \alpha_n k_n }{ 1 + b g } \\
0 & 0  &  -d_n - \dfrac{ k_n  }{ 1 + b g } 
\end{array}
\right).
\end{equation}
The three eigenvalues of the matrix $J_1$ are $-a$, $-d_e - \dfrac{\gamma_e}{b}$, and $-d_n - \dfrac{ k_n  }{ 1 + b g }$. Hence, the equilibrium $(\dfrac{1}{b}, 0, 0)$ is asymptotically stable.
	
Next, we verify the global attractiveness. From the third equation of \eqref{eq:model_1_n} and $d_n > 0$, there exists a finite time $t_3 >0$ such that
\begin{equation}
\label{eq:model_1_n_1}
N(t) < \dfrac{\epsilon}{ \alpha_n k_n }
\end{equation} 
for any $\epsilon >0$ and $t \geq t_3$. Thus, from the second equation of \eqref{eq:model_1_n}, 
\begin{equation}
\label{eq:model_1_n_2}
	\dfrac{{\rm d} E}{{\rm d} t} < - d_e E + \epsilon \Big(  \dfrac{T}{ T + g } \Big) - \gamma_e E T < - d_e E + \epsilon,
	\end{equation} 
	we have
	\begin{equation}
	\limsup_{\substack t \rightarrow \infty} E(t) \leq \dfrac{\epsilon}{d_e}.
	\end{equation}
	Hence, for any $\epsilon_1 > 0$, there exists $t_4 \geq t_3$ such that
	\begin{equation}
	\label{eq:model_1_n_3}
	E(t) < \dfrac{\epsilon + \epsilon_1}{ d_e }
	\end{equation} 
	for any $t \geq t_4$. Since $\epsilon >0$ and $\epsilon_1 >0$ are sufficiently small, we can choose $\epsilon$ and $\epsilon_1 $ such that $\gamma_c (\epsilon + \epsilon_1  ) < d_e$, namely, $\dfrac{\gamma_c (\epsilon + \epsilon_1  )  }{ d_e}< 1$. Hence, 
	\begin{equation}
	\label{eq:model_1_n_4}
	\dfrac{{\rm d} T}{{\rm d} t} \geq a T( 1 -bT) - \dfrac{\gamma_c (\epsilon + \epsilon_1  )  }{ d_e} T = aT \Big(  1- \dfrac{\gamma_c (\epsilon + \epsilon_1  )  }{ d_e} - bT \Big).
	\end{equation} 
	Applying the comparison principle, we have
	\begin{equation}
	\liminf_{\substack t \rightarrow \infty} T(t) \geq \dfrac{1}{b} \Big(  1- \dfrac{\gamma_c (\epsilon + \epsilon_1  )  }{ d_e} \Big).
	\end{equation}
	Since both $\epsilon >0$ and $\epsilon_1 >0$ are sufficiently small, we have
	\begin{equation}
	\liminf_{\substack t \rightarrow \infty} T(t) \geq \dfrac{1}{b}.
	\end{equation}
	From \eqref{eq:T_ub}, we have $T(t) \rightarrow \dfrac{1}{b}$ as $t \rightarrow \infty$. Hence, $(\dfrac{1}{b}, 0, 0)$ is globally attractive in $\Omega$.
\end{proof}

\begin{remark}
\label{rem:2}
From Theorem \ref{thm2.aaa} (i), any trajectory with initial conditions $T(0) > 0$, $E(0) >0$, and $N(0) > 0$ will move towards a state with cancer cells when there is no immune cell injection, i.e., $s_e = s_n = 0$. 
\end{remark}

From Theorem \ref{thm2.2}, if $a < \dfrac{ \gamma_c s_e}{d_e + \gamma_e/b}$, the cancer cell population will approach zero as time progresses. However, the time to achieve a given positive threshold is interesting in estimating the process of cancer evolution. This issue is discussed below.

\subsection{The first time the cancer cell population reaches a given level}
Here, we investigate the time at which the concentration of cancer cells reaches a specified level, given the initial conditions $T(0)=x$, $E(0)=y$, and $N(0)=z$. We define $\tau$ as the first time $T(t)$ equals a given positive value $\epsilon$. Namely, 
\begin{equation}
\label{eq:ft_1}
\tau( x, y, z) := \inf \{ t>0 \mid T(t)=\epsilon \}
\end{equation}
for $x, y, z \geq 0$. Here, $T(t)$, $E(t)$, and $N(t)$ represent the solutions of equation \eqref{eq:model_1} with the initial condition \\ $(T(0), E(0), N(0)) = (x, y, z)$. We may also write the solutions as $T(t, x, y, z)$, $E(t, x, y, z)$, and $N(t, x, y, z)$ in order to specify the initial conditions.

\begin{theorem}
\label{thm2.first}
Consider the equation \eqref{eq:model_1}, let $T(t)$, $E(t)$, and $N(t)$ be the solution of \eqref{eq:model_1} with initial condition $(T(0), E(0), N(0))=(x, y, z)$. If $a < \dfrac{\gamma_c s_e}{d_e + \gamma_e/b}$, for any $\epsilon \in (0, x)$, the function $\tau$ defined in \eqref{eq:ft_1} satisfies the following properties:
\begin{enumerate}
\item[\rm{(i)}] $\tau$ is continuous on $[\epsilon, +\infty) \times [\dfrac{s_e}{d_e + \gamma_e/b}, +\infty) \times [0, +\infty)$ and is smooth in $\Omega_1 = (\epsilon, +\infty) \times (\dfrac{s_e}{d_e + \gamma_e/b}, +\infty) \times (0, +\infty)$.
\item[\rm{(ii)}] $\tau$ satisfies the following partial differential equation:
\begin{equation}
\label{eq:ft_2}
\left\{
\begin{aligned}
&f_1(x, y, z ) \partial_{x} \tau + f_2(x, y, z ) \partial_{y} \tau + f_3(x, y, z )\partial_{z} \tau =  1, \quad (x, y, z ) \in \Omega_1, \\
&\tau(\epsilon, y, z ) =  0,   \quad ( y, z ) \in  [ \dfrac{s_e}{d_e + \gamma_e/b}, +\infty) \times [0, +\infty),
\end{aligned}
\right.
\end{equation}
where 
\begin{eqnarray*}
f_1(x, y, z ) &=& a x ( 1 - b x) - \gamma_c x y,\\
f_2(x, y, z ) &=& s_e - d_e y +
\alpha_n k_n z \Big(  \dfrac{x}{ x + g } \Big)  - \gamma_e x y,\\
f_3(x, y, z ) &=& s_n - d_n z - k_n z \Big(  \dfrac{x}{ x + g } \Big).
\end{eqnarray*}
\end{enumerate}
\end{theorem}

\begin{proof} The proof below follows the idea of Theorem 1.1 in Hynd et al. \citep{Hynd2022}.
	
(i). For a given $\epsilon > 0$, Theorem \ref{thm2.2} implies that there exists a finite time $t_5$ such that cancer cell concentration reaches $\epsilon$ at $t = t_5$, i.e., 
\begin{equation}
\label{eq:ft_3}
	T(t_5)=\epsilon.
\end{equation}

Let $(x^k, y^k, z^k)$ be a sequence so that $(x^k, y^k, z^k ) \rightarrow (x, y, z)$ as $k \rightarrow \infty$ and $x ^k \geq \epsilon$, $y ^k \geq \dfrac{s_e}{ d_e + \gamma_e/b}$, $z ^k \geq 0$. There exists a subsequence $\tau(x^{k_j}, y^{k_j}, z^{k_j})$ such that 
\begin{equation}
\label{eq:ft_4}
\tau_1= \liminf_{\substack k \rightarrow \infty} \tau(x^k, y^k, z^k ) =\lim_{\substack j \rightarrow \infty} \tau(x^{k_j}, y^{k_j}, z^{k_j}).
\end{equation}
From the definition of $\epsilon$,
\begin{equation}
\label{eq:ft_5}
T( \tau(x^k, y^k, z^k ),x^k, y^k, z^k )  =\epsilon 
\end{equation}
for any $k \in \mathbb{N}$. Because of the continuity of $T$, we have
\begin{equation}
\label{eq:ft_6}
T( \tau_1, x, y, z ) =\epsilon.
\end{equation}
For any $(x, y, z) \in [\epsilon, +\infty) \times [\dfrac{s_e}{d_e + \gamma_e/b}, +\infty) \times [0, +\infty)$, we have $x \geq \epsilon$. Hence,
\begin{equation}
\label{eq:ft_7}
\tau(x, y, z) \leq \tau_1 =\liminf_{\substack k \rightarrow \infty} \tau(x^k, y^k, z^k ) .
\end{equation}
There also exists another subsequence $\tau(x^{k_l}, y^{k_l}, z^{k_l} )$ such that 
\begin{equation}
\label{eq:ft_8}
\tau_2= \limsup_{\substack k \rightarrow \infty} \tau(x^k, y^k, z^k ) =\lim_{\substack l \rightarrow \infty} \tau(x^{k_l}, y^{k_l}, z^{k_l}).
\end{equation}
Since $y \geq \dfrac{s_e}{ d_e + \gamma_e/b}$, it follows from \eqref{eq:E0_low} that 
\begin{equation}
\label{eq:ft_9}
E(t) \geq \dfrac{s_e}{ d_e + \gamma_e/b}
\end{equation}
for all $t \geq 0$, which yields 
\begin{equation}
\label{eq:ft_10}
\dfrac{{\rm d} T }{{\rm d} t} <0
\end{equation}
for $0 \leq t \leq \tau_2$ and $x \geq \epsilon$. Thus 
\begin{equation}
\label{eq:ft_120}
\tau(x, y, z) = \tau_2 = \limsup_{\substack k \rightarrow \infty} \tau(x^k, y^k, z^k ).
\end{equation}
Therefore, \eqref{eq:ft_7} and \eqref{eq:ft_120} together give
\begin{equation}
\label{eq:ft_11}
\tau = \lim_{\substack k \rightarrow \infty} \tau(x^k, y^k, z^k ).
\end{equation}
Thus, $\tau$ is continuous on $[\epsilon, +\infty) \times [\dfrac{s_e}{d_e + \gamma_e/b}, +\infty) \times [0, +\infty)$. 
	
Let $x>\epsilon$, $y>\dfrac{s_e}{ d_e + \gamma_e/b}$, and $z>0$. From \eqref{eq:E0_low}, we have $E(t)>\dfrac{s_e}{ d_e + \gamma_e/b}$ for $\forall t \geq 0$. In addition, it follows from \eqref{eq:cancer_1} that $E(t)>0$ for $\forall t \geq 0$. Thus, noticing that $T(\tau, x, y, z) = \epsilon$, from the first equation of \eqref{eq:model_1}, and $a < \dfrac{\gamma_c s_e}{ d_e + \gamma_e/b} $, we have
\begin{equation}
\label{eq:thm_fta}
\dfrac{ {\rm d} T(\tau, x, y, z) }{ {\rm d} t} < ( a- \gamma_c E) T < \Big(  \dfrac{\gamma_c s_e}{ d_e + \gamma_e/b} - \dfrac{ \gamma_c s_e}{d_e + \gamma_e/b}  \Big) T =0.
\end{equation}
Since $T$ is smooth in a neighborhood of $( \tau, x, y, z)$, the implicit function theorem implies that $\tau$ is smooth in a neighborhood of $( x, y, z)$. Thus, $\tau$ is smooth in $\Omega_1$.
	
(ii).  We denote $T(t), E(t), N(t)$ as the solutions of system \eqref{eq:model_1} with initial conditions $T(0)=x> \epsilon$, $E(0)=y > \dfrac{s_e}{d_e + \gamma_e/b}$, and $N(0)=z>0$. For any $t \in ( 0, \tau(x, y, z ) )$, since the equation \eqref{eq:model_1} is an autonomous system, it is easy to obtain
\begin{equation}
\label{eq:ft_12}
\begin{aligned}
\tau(T(t), E(t), N(t) ) & = \inf\{ \tau_1 >0 \mid T(t+\tau_1)= \epsilon\} \\
& = \inf\{ \tau_1 >0 \mid T(t+\tau_1)= \epsilon ~  \mathrm{and} ~ \tau_1 > 0 \} \\
& = \inf\{ s-t \mid T(s)= \epsilon ~ \mathrm{and} ~ s > t \} \\
& = \inf\{ s \mid T(s)= \epsilon ~ \mathrm{and} ~ s > t \} - t \\
& = \inf\{ s  > t\mid T(s)= \epsilon \} - t \\
& = \tau(x, y, z) - t .
\end{aligned}
\end{equation}
Thus, we have
\begin{equation}
\label{eq:ft_13}
\begin{aligned}
\dfrac{ \tau(T(t), E(t), N(t) ) - \tau(x, y, z )}{t} = - 1 .
\end{aligned}
\end{equation}
Letting $t \rightarrow 0$, we have
\begin{equation}
\label{eq:ft_14}
\begin{aligned}
-1 &= \dfrac{{\rm d} }{{\rm d} t}  \tau(T(t), E(t), N(t) ) \Big\vert_{t=0}  \\
& = [ f_1(T, E, N ) \partial_{T} \tau + f_2(T, E, N) \partial_{E} \tau + f_3(T, E, N )\partial_{N} \tau ]\Big\vert_{t=0}  \\
& = f_1(x, y, z) \partial_{x} \tau + f_2(x, y, z) \partial_{y} \tau + f_3(x, y, z)\partial_{z} \tau.
\end{aligned}
\end{equation}
	
We show that $\tau(x, y, z)$ satisfies the boundary condition of equation \eqref{eq:ft_2}. Suppose that a function $\tau_*(x, y, z)$ satisfies \eqref{eq:ft_2}. Then, for $ 0< t< \tau(x, y, z ) $, we have
\begin{equation}
\label{eq:ft_15}
\begin{aligned}
 \dfrac{{\rm d} }{{\rm d} t}  \tau_*(T(t), E(t), N(t) ) =  f_1(T, E, N ) \partial_{T} \tau_* + f_2(T, E, N) \partial_{E} \tau_* + f_3(T, E, N )\partial_{N} \tau_*  =-1.
\end{aligned}
\end{equation}
Integrating both sides of equation \eqref{eq:ft_15} from $ t =0$ to $t = \tau(x, y, z ) $, we obtain
\begin{equation}
\label{eq:ft_16}
\tau_*( T(\tau(x, y, z ) ), E(\tau(x, y, z ) ) , N(\tau(x, y, z ) )  ) - \tau_*(x, y, z )  =  -\tau(x, y, z ) .
\end{equation}
Since  $T(\tau(x, y, z ) ) = \epsilon$, $E(\tau(x, y, z ) ) \geq \dfrac{s_e}{d_e + \gamma_e/b}$, and $\tau_*( \epsilon, E(\tau(x, y, z) ) , N(\tau(x, y, z ) )  ) =0$, $\tau_*(x, y, z ) = \tau(x, y, z ) $. This proves the theorem.
\end{proof}

From the proof, the equation \eqref{eq:ft_2} has a unique solution, which gives the function $\tau(x, y, z)$. According to the theory of the first-order partial differential equation (PDE), if we can find three independent prime integrals of the characteristic equation of the PDE \eqref{eq:ft_2}, i.e., the following equation:
\begin{equation}
\label{eq:ft_aa}
\begin{aligned}
\dfrac{ {\rm d} x }{  f_1(x, y, z) } = \dfrac{ {\rm d} y }{ f_2(x, y, z) } = \dfrac{ {\rm d} z }{ f_3(x, y, z) } = \dfrac{ {\rm d} \tau }{  1 },
\end{aligned}
\end{equation}
we can then obtain the solution of \eqref{eq:ft_2}. Unfortunately, since the coefficients of the PDE \eqref{eq:ft_2} are coupled and complex, it is not easy to obtain the expression of the solution for \eqref{eq:ft_2}. This study investigates the numerical approximation of $\tau$ in the section \ref{sec:sec1}.

\begin{remark}
\label{rem:3}
Theorem \ref{thm2.first} shows the first time at which the concentration of cancer cells $T$ reaches a given threshold. This result is significant for clinical applications for predicting the critical time point $\tau$, and the cancer cell concentration will remain at a low level when $t > \tau$.
\end{remark}

Furthermore, the following theorem gives an upper bound on the time point $\tau$.
\begin{theorem}
\label{thm2c}
If $a < \dfrac{\gamma_c s_e}{d_e + \gamma_e/b}$, for any $x \geq \epsilon > 0$, $y \geq \dfrac{s_e}{d_e + \gamma_e/b}$ and $z \geq 0$, we have
\begin{equation}
\label{eq:ft_17}
\begin{aligned}
\tau(x, y, z ) \leq \dfrac{ \ln \dfrac{x}{\epsilon} } {\dfrac{\gamma_c s_e}{d_e + \gamma_e/b} -a + a b \epsilon }.
\end{aligned}
\end{equation}
\end{theorem}
\begin{proof} Let $(T(t), E(t), N(t))$ be the solution of equation \eqref{eq:model_1} with initial value $(x, y, z)$. From the first equation of \eqref{eq:model_1}, we have
$$
T(t) = x {\mathrm e}^{ \int_{ 0}^{t} \big( a ( 1- bT(s)) - \gamma_c E(s)  \big) {\rm d} s},
$$
which gives
\begin{equation*}
\epsilon = x {\mathrm e}^{ \int_{ 0}^{\tau(x,y,z)} \big( a ( 1- bT(s)) - \gamma_c E(s)  \big) {\rm d} s}.
\end{equation*}
We then have
\begin{equation}
\label{eq:ft_18}
\int_{ 0}^{\tau(x,y,z)} \big( \gamma_c E(s) -a + a bT(s)  \big) {\rm d} s = \ln \dfrac{x}{\epsilon}.
\end{equation}
From \eqref{eq:ft_9} and \eqref{eq:ft_18}, 
\begin{equation}
\label{eq:ft_19}
\int_{ 0}^{\tau(x,y,z)} \left( \dfrac{\gamma_c s_e}{d_e + \gamma_e/b} -a + a bT(s) \right) {\rm d} s \leq \ln \dfrac{x}{\epsilon}.
\end{equation}
Noting that $T(s) \geq \epsilon$ for $0 \leq s \leq \tau(x,y,z)$, we obtain
\begin{equation}
\label{eq:ft_20}
\int_{ 0}^{\tau(x,y,z)} \left( \dfrac{\gamma_c s_e}{d_e + \gamma_e/b} -a + a b \epsilon  \right) {\rm d} s \leq \int_{ 0}^{\tau(x,y,z)} \left( \dfrac{\gamma_c s_e}{d_e + \gamma_e/b} -a + a bT(s)  \right) {\rm d} s \leq \ln \dfrac{x}{\epsilon}.
\end{equation}
Since $a < \dfrac{\gamma_c s_e}{d_e + \gamma_e/b}$ and $\dfrac{\gamma_c s_e}{d_e + \gamma_e/b} -a + a b \epsilon > 0$, it follows from \eqref{eq:ft_20} that 
\begin{equation*}
\tau(x,y,z) \Big( \dfrac{\gamma_c s_e}{d_e + \gamma_e/b} -a + a b \epsilon  \Big)  \leq  \ln \dfrac{x}{\epsilon},
\end{equation*}
which implies \eqref{eq:ft_17}. The proof is complete.
\end{proof} 

\subsection{Numetical analysis}
\subsubsection{Method}
In this section, we perform numerical simulations to further explore the stability of the steady states. Unless otherwise specified, the default parameter values used in the simulations are listed in Table \ref{tab:table1}. Additionally, based on equations \eqref{eq:ps_1}, \eqref{eq:ps_2}, and \eqref{eq:ps_4}, the values of $E^*$ and $N^*$ for each equilibrium $(T^*, E^*, N^*)$ of equation \eqref{eq:model_1} can be uniquely determined by the value of $T^*$. Therefore, in the discussions that follow, we present the two-dimensional phase portraits in the $(T, E)$-plane rather than the three-dimensional $(T, E, N)$-space.

\subsubsection{Bifurcation with immune parameter $s_e$ or $s_n$}
We begin by examining the existence and stability of the steady states of system \eqref{eq:model_1} through numerical simulations. Following this, we provide a numerical approximation of the function $\tau$, which describes the first time the concentration of cancer cells reaches a specified level.

\begin{table}[htbp]
	\centering
	\caption{Default parameter values. }
	\begin{tabular}{lllll}
		\toprule
		Parameter                        & Value   & Unit    & Reference                   \\
		\midrule
		$a$                         & 0.18         & $ \rm {day}^{-1} $& \citep{Kuznetsov1994,Talkington2018}                                         \\
		$ b$               & 1/30000           & $ \rm  ( cells/ \mu l )^{-1}$           & \citep{Talkington2018}              \\
		$ \gamma_c$             & 0.005           & $ \rm ( cells/ \mu l )^{-1} \rm day^{-1}$           & \citep{Moore2004,Talkington2018}                 \\
		$ s_e$             &     5      & $ \rm ( cells / \mu l ) \rm day^{-1}$          &       $  \rm -$       \\
		$ d_e$            & 0.06          & $  \rm day^{-1}$           & \citep{Moore2004,Talkington2018}               \\
		$ \alpha_n$              & 100         & $  \rm -$           & \citep{Talkington2018}                 \\
		$ k_n$              & 0.001        & $   \rm day^{-1}$     & \citep{Moore2004}             \\
		$ g$             & 100        &  $ \rm cells / \mu l $    & \citep{Talkington2018}    \\
		$ \gamma_e$             & 0.0001          & $ \rm ( cells / \mu l )^{-1} \rm day^{-1}$   & \citep{Talkington2018}  \\
		$ s_n$            &     7      & $ \rm ( cells / \mu l ) \rm day^{-1}$  &      $  \rm -$    \\
		$ d_n$             & 0.04          & $  \rm day^{-1}$           & \citep{Moore2004,Talkington2018}                \\
		\bottomrule
	\end{tabular}
	\label{tab:table1}
\end{table}

Fig. \ref{fig:fig01}(a) shows the number and stability of positive steady states as parameter $s_n$ varies, while $s_e$ is set to $0$. From Fig. \ref{fig:fig01}(a), we observe that when $s_n$ is an appropriate value, there are bistable steady states characterized by high and low concentrations of cancer cells, corresponding to either disease or healthy states, respectively. 

Fig. \ref{fig:fig01}(b) displays the number and stability of nonzero steady states as the parameter $s_e$ varies, with $s_n$ fixed at $0$. In contrast to Fig. \ref{fig:fig01}(a), this scenario has at most two positive steady states. Notably, only the positive steady state with a high concentration of cancer cells is stable, suggesting that when $s_n=0$, a patient is likely to progress to the disease state. 

\begin{figure}[htbp]
	\centering
	\includegraphics[width=15 cm]{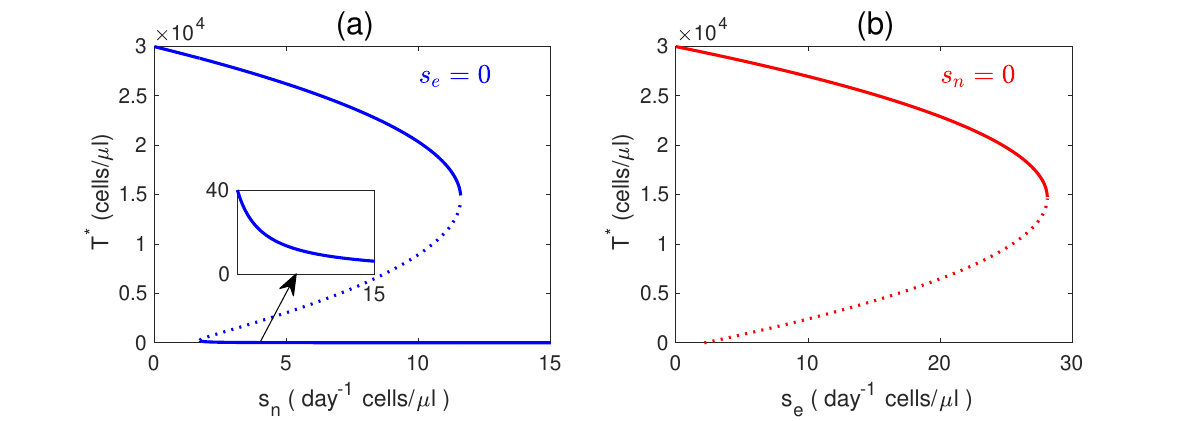}
	\caption{ \textbf{Stability of steady state with positive cancer cell concentration concerning changes in the parameter $s_n$ (or $s_e$) when $s_e =0$ (or $s_n =0$).}  The solid lines represent stable steady states, while the dotted lines indicate unstable steady states. All other parameter values are taken from Table \ref{tab:table1}. }
	\label{fig:fig01}
\end{figure}

Based on the default parameter values in Table \ref{tab:table1}, we investigate how varying the immune parameters $s_e$ and $s_n$ affect the number of positive steady states. According to Theorem \ref{thm2.b}, when $a > \dfrac{ \gamma_c s_e}{d_e}$, the system described by equation \eqref{eq:model_1} has at least one positive steady state. However, providing explicit conditions for the existence of more than one positive steady state is challenging. Therefore, we examine the number of positive steady states through numerical simulations. The number of positive steady states is determined by solving the equation defined by \eqref{eq:ps_4aa}. To illustrate our findings, we use different color regions: red for one positive steady state, green for two, and blue for three positive steady states. The results are presented in Fig. \ref{fig:fig1}. 

\begin{figure}[htbp]
	\centering
	\includegraphics[width=9.2 cm]{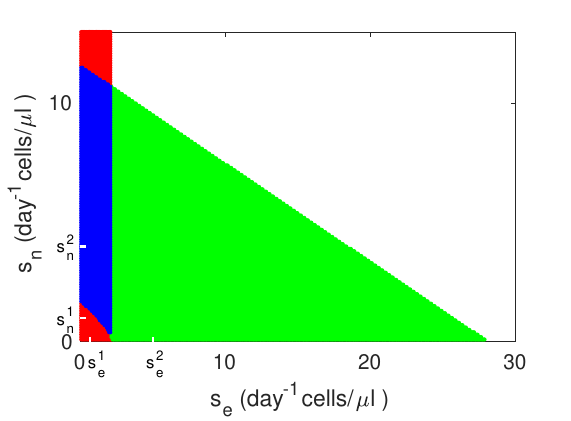}
	\caption{ \textbf{Dependence of the number of positive steady states on changes in parameters $(s_e, s_n)$.}  Different colored regions represent the number of positive steady states: red indicates one steady state, green indicates two, and blue indicates three. The empty area above the green region signifies that there is no positive steady state. The values $s_e^i$ and $s_n^i$ $(i=1,2)$ correspond to those used in Fig. \ref{fig:fig2}(a) and Fig. \ref{fig:fig3}(a). All other parameter values remain consistent with those in Table \ref{tab:table1}. }
	\label{fig:fig1}
\end{figure}

The results presented in Fig. \ref{fig:fig1} indicates that when the value of $s_e$ is small, for instance, $s_e=s_e^1$, equation \eqref{eq:model_1} may yield either one (shown in red) or three (shown in blue) positive steady states, depending on the value of $s_n$. This relationship is illustrated in Fig.\ref{fig:fig2}(a). As $s_e$ increases, for example to $s_e=s_e^2$, the equation may produce either zero or two positive steady states, again depending on the value of $s_n$, as shown in Fig. \ref{fig:fig2}(a). In Fig. \ref{fig:fig2}(b), we see eight trajectories representing the situation when there are two stable steady states. In this scenario, the states corresponding to either the tumor-free equilibrium or the positive equilibrium with a high concentration of cancer cells are stable, reflecting either a healthy or disease condition, as depicted in Fig. \ref{fig:fig2}(b). Furthermore, Fig. \ref{fig:fig2}(b) illustrates eight trajectories originating from different initial values, which converge either to the healthy state (the tumor-free equilibrium) or to the disease state (concentration by a high concentration of cancer cells).

\begin{figure}[htbp]
	\centering
	\includegraphics[width=15 cm]{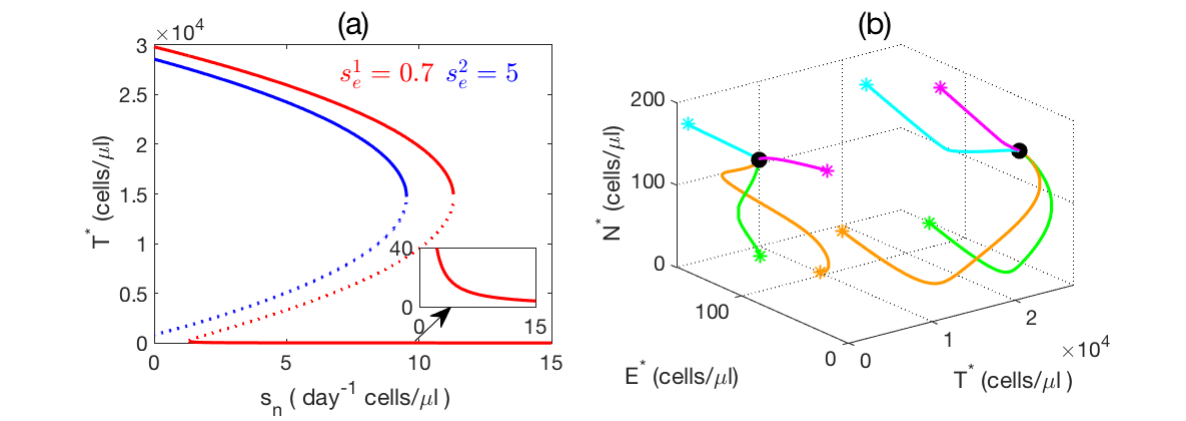}
	\caption{ \textbf{Existence and stability of steady states with changes in $s_n$.}  (a) Bifurcation diagram for positive steady states of the equation \eqref{eq:model_1} illustrating changes in $s_n$ and different values of $s_e$ (with $s_e^1 = 0.7$ and $s_e^2 = 5$). The solid lines indicate stable steady states, while the dotted lines represent unstable steady states. (b) Solutions of the system \eqref{eq:model_1} with different initial conditions, using parameters $(s_e, s_n) = (5, 7)$. The black dots indicate two stable steady states: the left state represents the tumor-free equilibrium, while the right state denotes the disease state. The eight lines illustrate trajectories converging towards either of the two stable steady states, with the corresponding initial conditions marked by stars. All other parameters remain consistent with those in Table \ref{tab:table1}. }
	\label{fig:fig2}
\end{figure}

\begin{figure}[htbp]
	\centering
	\includegraphics[width=15 cm]{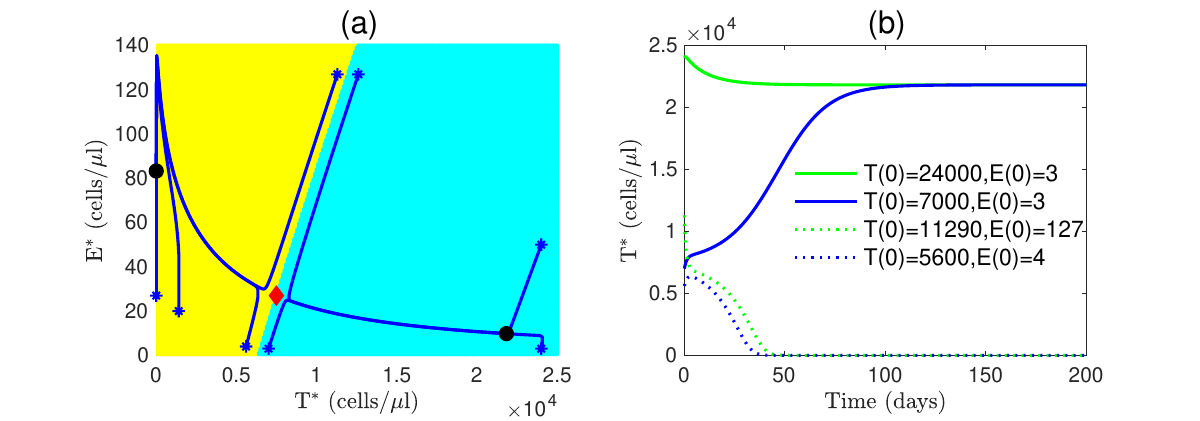}
	\caption{ \textbf{Illustration of the stability of steady states when $(s_e, s_n) = (5, 7)$.} (a) Basins of attraction for the steady states when the trajectory of system \eqref{eq:model_1} starts with an initial value of $N(0)= 170.79$. The red diamond dots represent unstable steady states, while the black dots indicate stable steady states. The eight blue lines illustrate trajectories that converge towards either of the stable steady states, with the corresponding initial states marked by blue stars.  (b) Time courses of tumor cells for the sample trajectories as they approach one of the two stable steady states, starting from different initial conditions. The parameters used are taken from Table \ref{tab:table1}. }
	\label{fig:fig20}
\end{figure}

Next, based on parameter values in Table \ref{tab:table1}, we analyze the attraction basins of different steady states to illustrate the long-term behaviors of three types of cells intuitively. With the initial values fixed at $N(0)= 170.79$, we consider initial values $(T(0), E(0))$ within the ranges $[100, 25000] \times [0.5, 140]$ using a uniform mesh. Understanding the long-term behaviors of the system \eqref{eq:model_1} under different initial conditions is crucial, as this knowledge can help us intervene in the system and protect against disease states. In Fig. \ref{fig:fig20}(a), the yellow region represents the attraction basins of the healthy state (the tumor-free equilibrium), while the cyan region represents the attraction basins of the disease state (characterized by a high concentration of cancer cells). The system \eqref{eq:model_1} can evolve into a healthy state if the initial values $(T(0), E(0) )$ are located in the yellow region. Conversely, if the initial values fall within the cyan region, the system \eqref{eq:model_1} may develop into a disease state.

We also analyze the stability of positive steady states and the two-dimensional phase portraits when the parameter $s_n$ is fixed. The results are illustrated in Fig. \ref{fig:fig3}. As seen in Fig. \ref{fig:fig1}, when the value of $s_n$ is small, for instance, $s_n=s_n^1$, equation \eqref{eq:model_1} may have one (red), two (green), or three (blue) positive steady states, depending on the value of $s_e$, as shown in Fig. \ref{fig:fig3}(a). Further increasing $s_n$, for example, to $s_n=s_n^2$, can result in three (blue), two (green), or zero positive steady states, again depending on the value of $s_e$, as illustrated in Fig. \ref{fig:fig3}(a). The two-dimensional phase portraits of equation \eqref{eq:model_1} are presented in Fig. \ref{fig:fig3}(b) for specific values of $(s_e, s_n) = (1.5, 4)$ where three positive steady states exist. In this scenario, bistable steady states arise, corresponding to high and low concentrations of cancer cells labeled as the disease and healthy states, respectively (Fig. \ref{fig:fig3}(b)). Moreover, Fig. \ref{fig:fig3}(b) illustrates eight trajectories with different initial values. These eight trajectories approach either the healthy state (low concentration of cancer cells) or the disease state (high concentration of cancer cells).

\begin{figure}[htbp]
	\centering
	\includegraphics[width=15 cm]{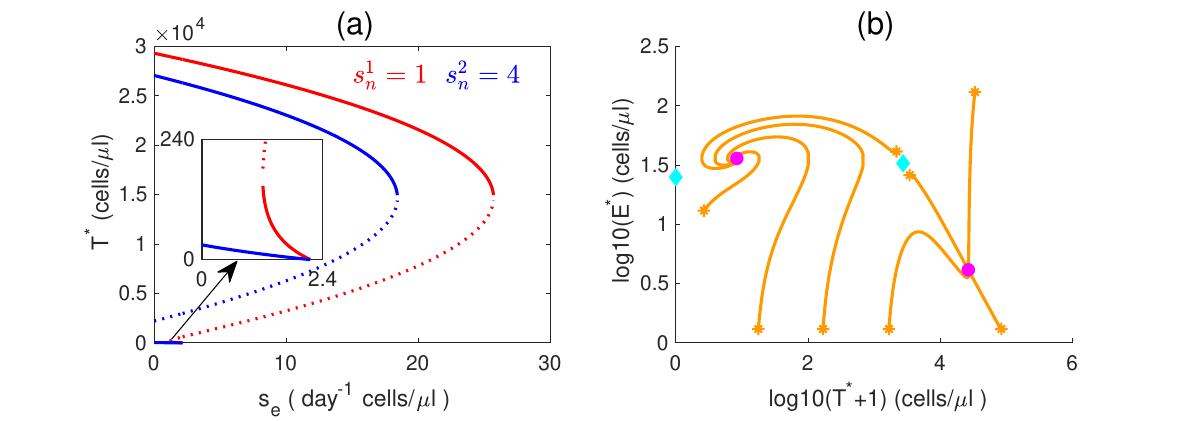}
	\caption{ \textbf{Existence and stability of steady states with changes in the parameter $s_e$.}   (a) Bifurcation diagram for positive steady states of equation \eqref{eq:model_1} with changes in $s_e$ and different values of $s_n$ (specifically, $s_n = 1$ or $s_n = 4$). The solid lines represent stable steady states, while the dotted lines indicate unstable steady states. (b) Phase portraits of the solutions for different initial conditions at $(s_e, s_n) = (1.5, 4)$. The two cyan diamonds represent unstable steady states, whereas the magenta dots indicate two stable steady states. The eight orange lines represent trajectories that converge towards either of the two stable steady states, with the corresponding initial conditions marked by orange stars. All other parameters remain consistent with those in Table \ref{tab:table1}. }
	\label{fig:fig3}
\end{figure}

\subsubsection{Numerical approximation of the function $\tau$}
\label{sec:sec1}

Additionally, setting $(s_e, s_n) = (1.5, 4)$ and using other parameters from Table \ref{tab:table1}, we perform numerical calculations to approximate $\tau(x, y, z)$ as defined in \eqref{eq:ft_1}. To ensure $y \geq \dfrac{s_e}{d_e + \gamma_e/b}$, we set $y \geq 1$, since $ \dfrac{s_e}{d_e + \gamma_e/b} \approx 0.49$. We create a uniform mesh $\{(x_i, y_j), i=1,..., N_1, j=1,..., N_2\}$ over the domain $(x, y) \in [ 1, 101] \times [1, 16]$. With a given threshold $\epsilon (=1)$, we examine the numerical solutions on the interval $[0, \hat{t}]$, where time steps are uniform and represented as $\{t_k, k=1,..., M\}$. The upper limit $\hat{t}$ is calculated as $\dfrac{ \ln \dfrac{x}{\epsilon} } {\dfrac{\gamma_c s_e}{d_e + \gamma_e/b} -a + a b \epsilon }$. For the first time $t_k$ such that $T(t_k) \leq \epsilon$, we take $t_k$ as the approximations for $\tau$. The numerical approximation of $\tau$ is shown in Fig. \ref{fig:tf} and Fig. \ref{fig:tf2}. When the initial concentrations of effector T cells and naive T cells are held constant, $\tau$ increases as the initial concentration of cancer cells rises, as shown in Fig. \ref{fig:tf2}(a). Conversely, when the initial concentrations of cancer cells and naive T cells are fixed, $\tau$ decreases as the initial concentration of effector T cells increases, which is depicted in Fig. \ref{fig:tf2}(b) increases.

\begin{figure}[htbp]
	\centering
	\includegraphics[width=9.2cm]{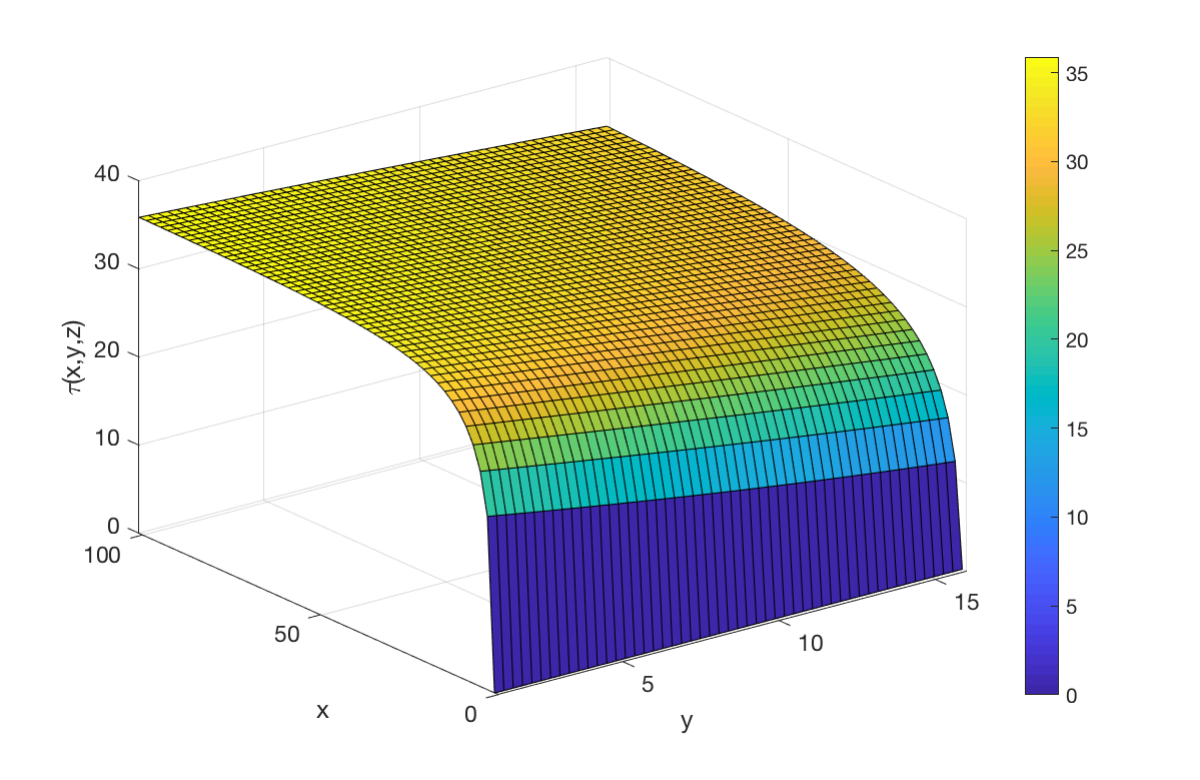}
	\caption{ \textbf{Numerical approximation of $\tau(x, y, z)$ defined in \eqref{eq:ft_1} for $(x, y) \in [ 1, 101] \times [1, 16]$ and $z =25$.} We set $a= 0.002$ while keeping the other parameters the same as those in Fig. \ref{fig:fig0}.}
	\label{fig:tf}
\end{figure}

\begin{figure}[htbp]
	\centering
	\includegraphics[width=15 cm]{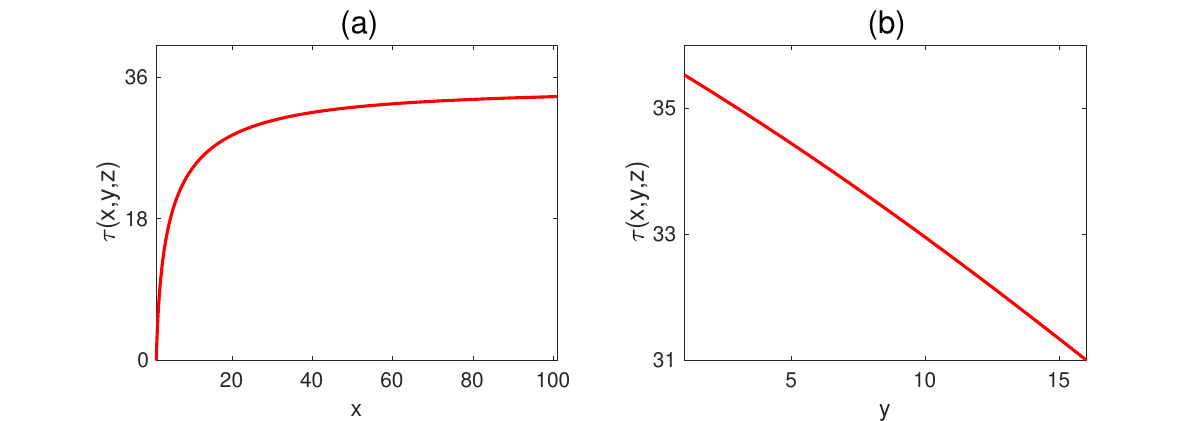}
	\caption{ \textbf{Numerical approximation of $\tau(x, y, z)$ defined in \eqref{eq:ft_1} with respective to $x$ or $y$.}  (a) Numerical approximation of $\tau(x, y, z)$ for $x \in [1, 101]$ with $ y = 11$. (b) Numerical approximation of $\tau(x, y, z)$ for $x = 61$ and $ y \in [1, 16]$. In both cases, we  keep $z =25$ and $a = 0.002$, while the other parameters remain the same as those in Fig. \ref{fig:fig0}.}
	\label{fig:tf2}
\end{figure}

\subsection{Optimal control problem for immunotherapy}

In this section, we consider the optimal treatment strategies for immunotherapy, focusing on the timing and sequence of drug administration.

To describe the immunotherapy protocol, we introduce the time-dependent control terms $u_e(t)$ and $u_n(t)$ into the model and discuss the optimal strategy for adiministering effector and naive T cells into the bloodstream. Thus, the model equation \eqref{eq:model_1} is rewritten as
\begin{equation}
\label{eq:model_1_op}
\left\{
\begin{aligned}
\dfrac{{\rm d} T}{{\rm d} t}& =  a  T ( 1 - b T) - \gamma_c E T,\\
\dfrac{{\rm d} E}{{\rm d} t}& =  u_e(t) s_e - d_e E + \alpha_n k_n N \Big(  \dfrac{T}{ T + g } \Big)  - \gamma_e E T, \\
\dfrac{{\rm d} N}{{\rm d} t}& =  u_n(t) s_n - d_n N - k_n N \Big(  \dfrac{T}{ T + g } \Big).
\end{aligned}
\right.
\end{equation} 
The control functions $u_e(t)$ and $u_n(t)$ are Lebesgue measurable and represent the proportions of the maximum amounts of $s_e$ and $s_n$, respectively. We have the constrains $0 \leq u_e(t) \leq 1$ and $0 \leq u_n(t) \leq 1$ for all $t \geq 0$. 

Let $[0, t_f]$ be the treatment interval; we define the objective functional for the treatment strategy $(u_e(t), u_n(t))$ as
\begin{equation}
\label{eq:3objective_bb}
J(u_e, u_n) = T(t_f) + \int_0^{t_f} \big( a_1 T(t) +  a_2  u_e(t) +  a_3 u_n(t) \big) {\rm d} t.
\end{equation}
The positive constants $a_1$, $a_2$, and $a_3$ are weight factors that balance the four contributions. In the objective functional \eqref{eq:3objective_bb}, $T(t_f)$ represents the concentration of cancer cells at the end of the treatment. The three components in the integrand represent the accumulated concentration of cancer cells during treatment, as well as the costs associated with the treatments $u_e$ and $u_n$ treatments. We aim to minimize $J(u_e, u_n)$ to balance the tumor burden against the treatment costs. 

The optimal control problem can be described as
\begin{equation}
\label{eq:opt1}
\begin{aligned}
&\min_{( u_e, u_n)\in\mathcal{U}} J( u_e, u_n ) = T(t_f) + \int_0^{t_f} \big( a_1 T(t) +  a_2  u_e(t) +  a_3 u_n(t) \big)  {\rm d} t,\\
&\mathrm{subject\ to}\ \  \left\{
\begin{aligned}
\dfrac{{\rm d} T}{{\rm d} t}& =  a T ( 1 - b T) - \gamma_c E T,\\
\dfrac{{\rm d} E}{{\rm d} t}& =  u_e s_e - d_e E +
\alpha_n k_n N \Big(  \dfrac{T}{ T + g } \Big)  - \gamma_e E T,   \\
\dfrac{{\rm d} N}{{\rm d} t}& =  u_n s_n - d_n N - k_n N \Big(  \dfrac{T}{ T + g } \Big), \\
T(0) & = x, E(0) = y, N(0) = z,
\end{aligned}
\right.
\end{aligned}
\end{equation}
with the set of admissible controls
\begin{equation}
\label{eq:2_control_set}
\mathcal{U} = \big \{ (u_e, u_n) \in L^{ \infty}([0, t_f], \mathbb{R}^2 ) \mid 0 \leq u_e(t) \leq 1, 0 \leq u_n(t) \leq 1, \forall t \in [0, t_f] \big \}.
\end{equation}

According to optimal control theory \citep{Fleming1975,Bressan2007,Ledzewicz2012book,Agrachev2004}, there exists a pair $(\tilde{u}_{e}, \tilde{u}_{n})\in \mathcal{U}$ that solves the optimal control problem \eqref{eq:opt1}. By applying Pontryagin's minimum principle (PMP), we can derive the necessary conditions for optimal control solutions.  

The Hamiltonian of the problem \eqref{eq:opt1} is defined as
\begin{equation}
\label{eq:hami}
\begin{aligned}
& H(t, T, E, N, u_e, u_n, \lambda_1, \lambda_2, \lambda_3 )   \\
=\ & a_1 T +  a_2  u_e +  a_3 u_n  + \lambda_1 \big(  a T ( 1 - b T) - \gamma_c E T \big)   \\
&{} + \lambda_2 \left(   u_e s_e - d_e E +
\alpha_n k_n N \Big(  \dfrac{T}{ T + g } \Big)  - \gamma_e E T  \right) \\
&{} + \lambda_3 \left(   u_n s_n - d_n N - k_n N \Big(  \dfrac{T}{ T + g } \Big)   \right).
\end{aligned}
\end{equation}
According to the PMP, if \eqref{eq:opt1} has an optimal control solution $(\tilde{u}_{e}, \tilde{u}_{n})\in \mathcal{U}$, and if $(\tilde{T}, \tilde{E}, \tilde{N})$ is the corresponding solution of \eqref{eq:model_1} over the interval $[0, t_f]$, there exist absolutely continuous adjoint variables $\lambda_i (i = 1, 2, 3)$ that satisfy the following equation:
\begin{equation}
\label{eq:model_1_af}
\left\{
\begin{aligned}
\dfrac{{\rm d} \lambda_1}{{\rm d} t} &= {}\left.- \dfrac{ \partial H}{ \partial T}\right\vert_{(\tilde{T}, \tilde{E}, \tilde{N})}\\
& = {} -a_1 - \lambda_1 (a - 2 a b \tilde{T} - \gamma_c \tilde{E} ) - \lambda_2 \Big( \dfrac{g  \alpha_n k_n \tilde{N} }{ (\tilde{T} + g)^2 } - \gamma_e \tilde{E} \Big)  + \lambda_3  \dfrac{g k_n \tilde{N} }{ (\tilde{T} + g)^2 },\\
\dfrac{{\rm d} \lambda_2}{{\rm d} t} &= {}\left.- \dfrac{ \partial H}{ \partial E}\right\vert_{(\tilde{T}, \tilde{E}, \tilde{N})}\\
& ={} \gamma_c \lambda_1 \tilde{T} + \lambda_2( d_e + \gamma_e \tilde{T}),   \\
\dfrac{{\rm d} \lambda_3}{{\rm d} t} &= {}\left.- \dfrac{ \partial H}{ \partial N}\right\vert_{(\tilde{T}, \tilde{E}, \tilde{N})}\\
& ={} -\lambda_2  \dfrac{  \alpha_n k_n \tilde{T} }{ \tilde{T} + g } + \lambda_3 (d_n + \dfrac{ k_n \tilde{T} }{ \tilde{T} + g } ),
\end{aligned}
\right.
\end{equation} 
with conditions $\lambda_1(t_f) = 1$, $\lambda_2(t_f) = 0$ and $\lambda_3(t_f) = 0$.

Since the Hamiltonian $H$ is linear in the controls $u_e$ and $u_n$, the minimizations of $H$ can be decoupled and be performed separately. We define the \emph{switching functions} $ \phi_1 = a_2 + s_e \lambda_2 $ and $\phi_2 = a_3 + s_n \lambda_3$ according to the coefficients of the controls $u_e$ and $u_n$ in equation \eqref{eq:hami}, respectively. The optimal controls $\tilde{u}_{e}$ and $\tilde{u}_{n}$ are given respectively by
\begin{equation}
\label{eq:3bb}
\tilde{u}_{e}(t) = \left\{
\begin{array}{ll}
0, &\phi_1(t) > 0,\\
1, &\phi_1(t) < 0 \\
\end{array}
\right.
\end{equation}
and 
\begin{equation}
\label{eq:3bb1}
\tilde{u}_{n}(t) = \left\{
\begin{array}{ll}
0, &\phi_2(t) > 0,\\
1, &\phi_2(t) < 0 .\\
\end{array}
\right.
\end{equation}

When the switching function is nonzero, bang–bang is obtained from \eqref{eq:3bb} and \eqref{eq:3bb1}. Bang-bang control suggests an ON-OFF treatment in clinical applications, which implies a drug holiday when $\tilde{u}_{e}(t)  = 0$ and $\tilde{u}_{n}(t) =0$, and a maximum drug dose when $\tilde{u}_{e}(t) =1$ or $\tilde{u}_{n}(t) =1$.

It is important to note that the PMP only provides a necessary condition for solving the optimal control problem. Therefore, the Bang-bang strategy, outlined in \eqref{eq:3bb} and \eqref{eq:3bb1}, is a candidate for the optimal control. However, this strategy is based on the goal of minimizing the Hamiltonian $H$, which in turn minimizes the objective functional $J(u_e, u_n)$ for $(u_e, u_n)\in \mathcal{U}$.

From a biological perspective, the Band-bang strategy suggests a treatment protocol in which the maximum dose is administered to reduce the tumor burden, followed by the drug holiday once the burden has decreased to a low level. This ON-OFF treatment approach is meaningful as it helps to balance the tumor burden with the associated treatment costs.

We can numerically explore the optimal immunotherapeutic protocols using an iterative forward-backward sweep method (FBSM) \citep{LenhartandWorkman2007}. Iinitially, the boundary values for the state variables $(T, E, N)$ and the adjoint variables $(\lambda_1, \lambda_2, \lambda_3)$ are given at $t=0$ and $t=t_f$, respectively. The iterative scheme is implemented following the numerical scheme outlined below. 

First, an initial guess is made for $(u_e, u_n)$. Next, based on equation \eqref{eq:opt1}, the state variables are solved forward in time using the initial conditions $(T(0), E(0), N(0)) = (T_0, E_0, N_0)$ along with $(u_e, u_n)$. Afterward, the adjoint variables are solved backward in time according to the system described in \eqref{eq:model_1_af}, utilizing the state variables $(u_e, u_n)$ and the boundary conditions $(\lambda_1(t_f), \lambda_2(t_f), \lambda_3(t_f)) = (1, 0, 0)$. Finally, the controls are updated based on equations \eqref{eq:3bb} and \eqref{eq:3bb1}, using the new state and adjoint variables. This process continues until the convergence condition is met. Throughout this procedure, we employ a fourth-order Runge-Kutta method with uniform discretization in time to solve the equations \eqref{eq:model_1_op} and \eqref{eq:model_1_af}. 

In the subsequent discussions, the end of the treatment interval is set as $t_f=60 \ {\rm day}$ with a time step of $\Delta t = 0.006$ for the model simulations. The initial conditions are specified as $T(0) = 10000 \ {\rm cells / \mu l}$, $E(0) = 120  \ {\rm cells / \mu l}$, and $N(0) = 140  \ {\rm cells / \mu l}$. The results are presented in Fig. \ref{fig:fig4}, Fig. \ref{fig:fig6}, and Fig. \ref{fig:fig5}. 

\begin{figure}[htbp]
	\centering
	\includegraphics[width=15 cm]{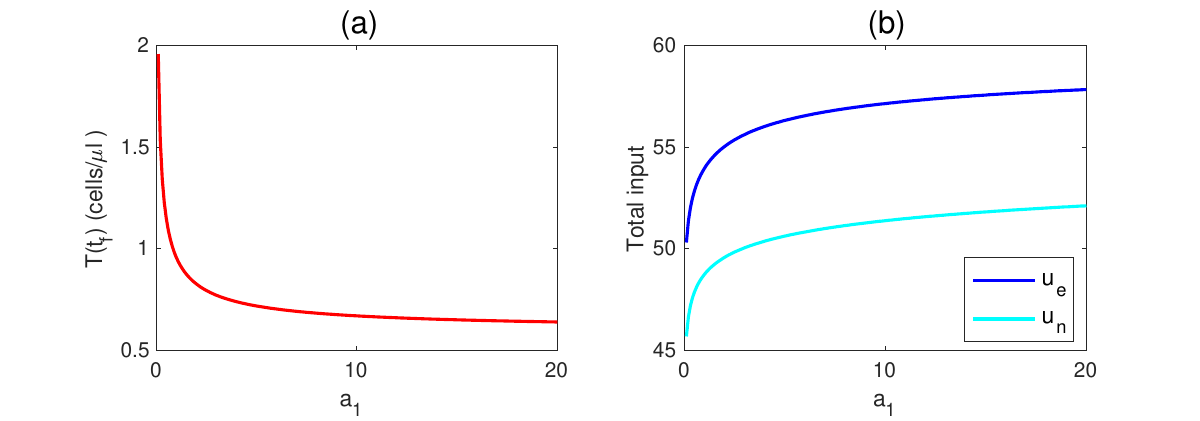}
	\caption{ \textbf{Optimal control solutions for $a_1 \in [0.1, 20]$.} (a) The concentration of cancer cells at $t=t_f$ for various values of $a_1$. (b) Total control amount of control $u_e$ (blue) and $u_n$ (cyan) administered during the treatment interval $[0, t_f]$ for different $a_1$ value. We set $a_2=a_3 = 1$ in these simulations and used the other parameter values provided in Table \ref{tab:table1}. }
	\label{fig:fig4}
\end{figure}

Fig. \ref{fig:fig4} illustrates the concentration of cancer cells $T(t_f)$ at the end of the treatment interval, as well as the cumulative number of the optimal controls $(u_e, u_n)$, with changes in the weight parameter $a_1$. From Fig. \ref{fig:fig4}(a), we observe that the concentration of cancer cells significantly decreases with the increase of $a_1$ when $a_1$ is small. However, as $a_1$ increases and becomes larger, the sensitivity to changes in $a_1$ diminishes. Additionally, the total amount of $u_e$ and $u_n$ increases with $a_1$ (Fig. \ref{fig:fig4}(b)). 

Fig. \ref{fig:fig6} shows the concentration of cancer cells $T(t_f)$ at the end of the treatment interval, along with the cumulative number of the optimal controls $(u_e, u_n)$, as the weight parameter $a_2$ changes (with the assumption that $a_3 = a_2$). In Fig. \ref{fig:fig6}a, it can be seen that the concentration of cancer cells $T(t_f)$ increases with $a_2$, while the total amount of $u_e$ and $u_n$ decreases as $a_2$ increases (see Fig. \ref{fig:fig6}b).

\begin{figure}[htbp]
	\centering
	\includegraphics[width=15 cm]{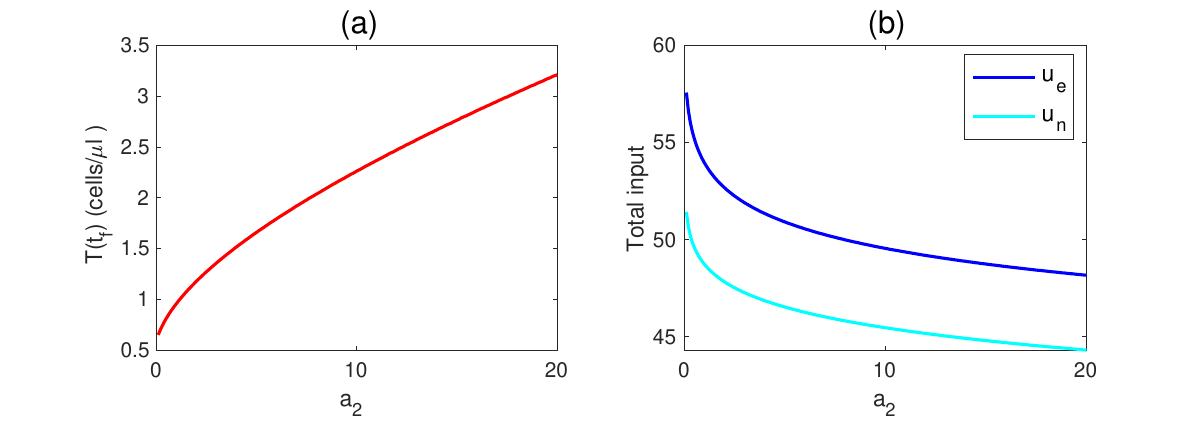}
	\caption{ \textbf{Optimal control dynamics for $a_2$ (with the assumption $a_3 = a_2) \in [0.1, 20]$.}   (a) Time courses of cancer cell concentration for varying values of $a_2$. (b) Time courses of total control input $u_e$ (in blue) and $u_n$ (in cyan) for different values of $a_2$. Here, $a_1$ is set to $1$, and other parameter values are taken from Table \ref{tab:table1}. }
	\label{fig:fig6}
\end{figure}

Fig. \ref{fig:fig5} compares the optimal control method with a constant treatment approach. According to Theorem \ref{thm2.aaa}\rm{(i)}, in the absence of immunotherapy, the concentration of cancer cells will rise to a high level, reaching a steady state solely by cancer cell population. As shown in Fig. \ref{fig:fig5}(a), the optimal treatment policy significantly reduces the cancer cell population compared to the constant treatment policy. Both treatment strategies involve the same total drug dosage, as illustrated in Fig. \ref{fig:fig5}(b). This outcome demonstrates that the optimal treatment policy achieves better results with the same amount of drugs. 

\begin{figure}[htbp]
	\centering
	\includegraphics[width=15 cm]{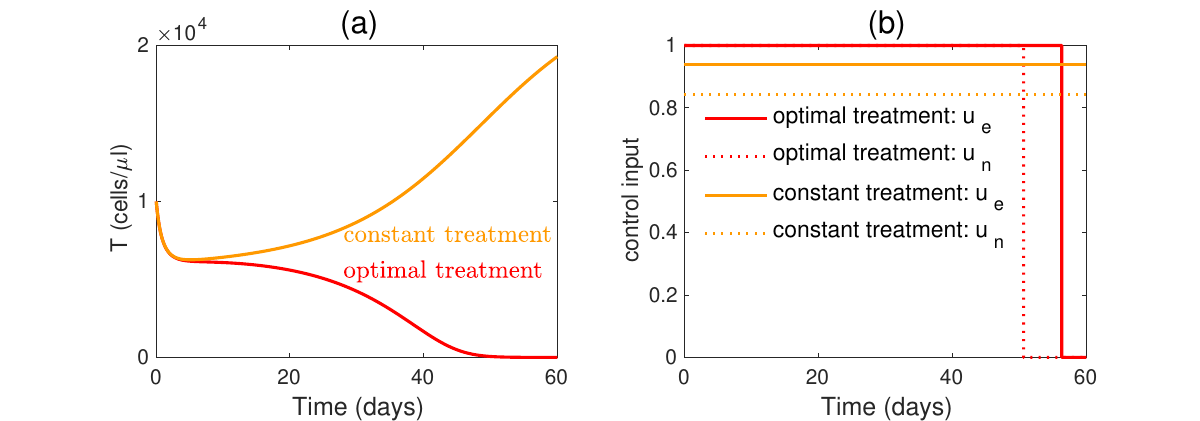}
	\caption{ \textbf{Comparison of optimal and constant treatment dynamics.} (a) Time courses of the cancer cell concentration with optimal (red) and constant (orange) treatment policies.  (b) Time courses of the control strategy $(u_e, u_n)$ with optimal (red) and constant (orange) treatment policies. Here, we set $a_1 = 5$ and $a_2 = a_3 = 1$. Other parameter values are taken from Table \ref{tab:table1}. }
	\label{fig:fig5}
\end{figure}

\section{Discussion}
\label{sec:4}
Understanding the interactions between tumor and immune cells is crucial in the fight against cancer. In this context, a nonlinear differential equation model (equation \eqref{eq:model_1}) proposed by \citep{Talkington2018} is used to explore the interactions between chronic myelogenous leukemia (CML) cells and two different types of immune cells, as well as potential immunotherapy strategies. The positivity and boundaries of the solutions are first examined. The effects of immune cell production rates $s_e$ and $s_n$ on the immune system are then analyzed. Appropriate values for $s_e$ and $s_n$ result in bistable steady states. Furthermore, the stability of positive steady states under different conditions is discussed. Through the model, we can determine the first time $\tau$ at which the cancer cell concentration reaches the desired level as a function of initial values. This function satisfies a partial differential equation with appropriate boundary conditions. The upper bound on the time $\tau$ is estimated based on model parameters.

Our study employed optimal control methods to identify the best timing and sequence for a therapeutic protocol. We analyzed the optimal combination therapy strategies using the Pontryagin Maximum Principle (PMP) approach. Our numerical simulations show that optimal treatment can reduce cancer cell concentrations more effectively than continuous treatment, even when the same treatment interval and total control input are used. We obtain a bang-bang treatment strategy from the optimal control problem, as illustrated in Fig. \ref{fig:fig5}.

This paper focuses on analyzing a conceptual model that explores the dynamic interactions between tumor cells and the immune system. However, model \eqref{eq:model_1} does not incorporate clinical data. Therefore, further work is required to integrate patient-specific information into mathematical models, which could lead to the development of more effective, personalized therapeutic strategies. It is well-established that certain chemotherapies can enhance tumor immunity, and combining immunotherapy with chemotherapy holds great potential for therapeutic synergy in cancer treatment \citep{Emens2015,Lake2005}. Consequently, investigating the combination of immunotherapy and chemotherapy protocols for cancer patients through mathematical modeling would be a valuable avenue for future research.

\section*{Declaration of Competing Interest}
The authors declare there is no conflict of interest.

\section*{Acknowledgments}
This work was funded by the National Natural Science Foundation of China (NSFC 12331018). 

\bibliographystyle{elsarticle-num.bst}
\bibliography{papers.bib}

\end{document}